\newif\ifconf

\ifconf
	\documentclass[compsoc,conference,a4paper,10pt,times]{IEEEtran}
	\usepackage{times}
	\makeatletter
	\renewcommand\paragraph{\@startsection{paragraph}{4}{\z@}%
	{0.5ex \@plus0.2ex \@minus.1ex}
	{0em}%
  	{\normalfont\normalsize\bfseries}} 
	\makeatother
\else
	\documentclass[11pt,a4paper]{article}
	\usepackage{fullpage}
	\usepackage{appendix}
\fi

\usepackage{amsmath,amsfonts,amssymb,amstext,dsfont,mathrsfs,amsthm}
\usepackage{mathtools}

\usepackage{enumerate}
\usepackage{csquotes}
\usepackage{verbatim}
\usepackage{stfloats}

\usepackage[dvipsnames]{xcolor}
\usepackage{breakcites}
\usepackage{mdframed}
\usepackage{algorithm2e}
\usepackage{nicefrac}

\usepackage{enumitem}

\usepackage{graphicx}
\usepackage{caption}
\usepackage{subcaption}
\usepackage{float}

\usepackage{bm}

\allowdisplaybreaks
\usepackage{multirow}
\usepackage{booktabs}
\usepackage{url}
\usepackage[colorlinks=true,citecolor=ForestGreen,linkcolor=NavyBlue]{hyperref}

\usepackage{cleveref} 
\usepackage{placeins}

\usepackage{tikz}
\usetikzlibrary{shapes,backgrounds}

\title{
	\bfseries The Coding Limits of Robust Watermarking for Generative Models
}

\ifconf
\author{
\IEEEauthorblockN{
Danilo Francati\IEEEauthorrefmark{1}\IEEEauthorrefmark{4},
Yevin Nikhel Goonatilake\IEEEauthorrefmark{2}\IEEEauthorrefmark{4},
Shubham Vivek Pawar\IEEEauthorrefmark{3}\IEEEauthorrefmark{4},\\
Daniele Venturi\IEEEauthorrefmark{1},
Giuseppe Ateniese\IEEEauthorrefmark{2}
}
\\
\IEEEauthorblockA{\IEEEauthorrefmark{1}
Sapienza University of Rome, Rome, Italy\\
\texttt{\{francati, venturi\}@di.uniroma1.it}}
\IEEEauthorblockA{\IEEEauthorrefmark{2}
George Mason University, Fairfax, USA\\
\texttt{\{ygoonati, ateniese\}@gmu.edu}}
\IEEEauthorblockA{\IEEEauthorrefmark{3}
Royal Holloway, University of London, Egham, UK\\
\texttt{shubham.pawar.2022@live.rhul.ac.uk}}
\\
\IEEEauthorblockA{\IEEEauthorrefmark{4}Equal contribution.}
}

\else
\author{
Danilo Francati\thanks{Equal contribution.} \\
Sapienza University of Rome \\
\texttt{francati@di.uniroma1.it}
\and
Yevin Nikhel Goonatilake\footnotemark[1] \\
George Mason University \\
\texttt{ygoonati@gmu.edu}
\and
Shubham Vivek Pawar\footnotemark[1] \\
Royal Holloway, University of London \\
\texttt{shubham.pawar.2022@live.rhul.ac.uk}
\and
Daniele Venturi \\
Sapienza University of Rome \\
\texttt{venturi@di.uniroma1.it}
\and
Giuseppe Ateniese \\
George Mason University \\
\texttt{ateniese@gmu.edu}
}

\date{}
\fi

\newcommand{\NN}{\mathbb{N}}

\newcommand{\cX}{\mathcal{X}}
\newcommand{\cY}{\mathcal{Y}}

\newcommand{\cK}{\mathcal{K}}

\newcommand{\cF}{\mathcal{F}}

\newcommand{\bin}{\{0,1\}}
\newcommand{\supp}{\mathrm{supp}}

\newcommand{\cModel}{\mathcal{M}}

\newcommand{\poly}{\textsf{poly}}

\newcommand{\blue}[1]{{\color{black}#1}}

\newtheorem{thm}{Theorem}

\newtheorem{coro}{Corollary}
\theoremstyle{definition}

\newtheorem{defn}{Definition}

\newtheorem{construction}{Construction}

\newcommand{\bits}{\{0, 1\}}

\newcommand{\KGen}{\mathsf{KGen}}
\newcommand{\Encode}{\mathsf{Enc}}
\newcommand{\Decode}{\mathsf{Dec}}

\let\originalleft\left
\let\originalright\right
\renewcommand{\left}{\mathopen{}\mathclose\bgroup\originalleft}
\renewcommand{\right}{\aftergroup\egroup\originalright}

\newcommand{\sk}{\mathit{sk}}

\newcommand{\negl}{\mathsf{negl}}

\newcommand{\secpar}{\lambda}
\newcommand{\secparam}{{1^\secpar}}

\newcommand{\getsr}{{\:{\leftarrow{\hspace*{-3pt}\raisebox{.75pt}{$\scriptscriptstyle\$$}}}\:}} 

\newcommand{\Prob}{\mathbb{P}}

\newcommand{\invalid}{\mathtt{invalid}}
\newcommand{\valid}{\mathtt{valid}}
\newcommand{\tampered}{\mathtt{tampered}}
\newcommand{\true}{\mathtt{true}}
\newcommand{\false}{\mathtt{false}}

\newcommand{\alphabet}{{\Sigma}}

\newcommand{\msg}{\mu}
\newcommand{\cdw}{\gamma}

\newcommand{\cdwlen}{n}
\newcommand{\thr}{t}

\newcommand{\dist}{\mathrm{dist}}

\newcommand{\LLM}{\mathsf{Model}}
\newcommand{\Watermark}{\mathsf{Watermark}}
\newcommand{\Detect}{\mathsf{Detect}}

\begin{document}

\maketitle

\begin{abstract}
We study a basic question about cryptographic watermarking for generative models: how reliable can a watermark remain when an adversary is allowed to corrupt the encoded signal? To address this question, we introduce a minimal coding abstraction that we call a \emph{zero-bit tamper-detection code}. This is a secret-key procedure that samples a pseudorandom codeword and, given a candidate word, decides whether it should be treated as unmarked content or as the result of tampering with a valid codeword. It captures the two core requirements of robust watermarking: \emph{soundness} and \emph{tamper detection}.

Within this abstraction we prove a sharp unconditional limit on robustness to independent symbol corruption. For an alphabet of size $q$, there is a critical corruption rate of $1 - 1/q$ such that no scheme with soundness, even relaxed to allow a fixed constant false positive probability on random content, can reliably detect tampering once an adversary can change more than this fraction of symbols. In particular, in the binary case no cryptographic watermark can remain robust if more than half of the encoded bits are modified. We also show that this threshold is tight by giving simple information-theoretic constructions that achieve soundness and tamper detection for all strictly smaller corruption rates.

We then test experimentally whether this limit appears in practice by looking at the recent watermarking for images of Gunn, Zhao, and Song (ICLR 2025). We show that a simple crop and resize operation reliably flipped about half of the latent signs and consistently prevented belief-propagation decoding from recovering the codeword, erasing the watermark while leaving the image visually intact.

\ifconf
\else
    \vspace{.3cm}
\fi
\end{abstract}

\ifconf
\else
\newpage
\fi

\section{Introduction}\label{sec:intro}

In recent years, generative AI models such as GPT, Llama, and Claude have made it increasingly difficult to tell human-produced content from machine-generated text, code, and images. This new reality raises a central question for both technology and society: how can we reliably distinguish what is AI-generated from what is authentically human? 

Watermarking is one of the most promising strategies for meeting this challenge~\cite{US23,EU24}. The idea is to embed a secret, imperceptible pattern into model outputs at generation time, enabling later verification by anyone with the appropriate key. Ideally, a watermark should be robust, surviving even substantial adversarial editing, while also being undetectable to those without the key and leaving the original content unchanged in quality or meaning~\cite{Aaronson22,Kirchenbauer_23,C:ChrGun24}.

Recent work by Christ and Gunn~\cite{C:ChrGun24} formalized cryptographic watermarking using pseudorandom error-correcting codes (PRCs), showing that such watermarks can be both undetectable and robust to significant amounts of tampering. But their results naturally lead to a deeper question: within such cryptographically grounded, coding-based frameworks, what are the true limits of robustness? Is there an unconditional barrier that no watermarking scheme of this form can overcome?

This question has been explored from two sides. On one side, the ``Watermarks in the Sand'' (WiTS) work~\cite{zhang2023watermarks,zhang2024watermarks} showed that any sufficiently robust watermark is, in principle, vulnerable: if an adversary has enough power, especially access to quality and perturbation oracles, it can eventually erase any watermark without degrading content quality. However, WiTS leaves open the concrete thresholds for robustness that efficient watermarking schemes can actually achieve. On the other side, PRC-based schemes demonstrate the power of current cryptographic techniques but do not establish whether their level of robustness is the best possible in the coding sense.

In this paper, we bring these perspectives together. We introduce a simple abstraction, the \emph{zero-bit tamper-detection code}, which we use as a minimal model for cryptographic watermarking based on pseudorandom codes. This notion can be viewed as the zero-bit specialization of a PRC equipped with an explicit tamper-detection predicate, and it captures the main coding issues that arise in such watermarking schemes. Using this notion, we prove a tight unconditional threshold on robustness to independent symbol corruption in this model: no scheme realizing a zero-bit tamper-detection code can reliably survive if more than half of the encoded bits are altered in the binary case. More generally, for a $q$-ary alphabet the corresponding limit is $(1-1/q)$ of the encoded symbols. Conversely, for any small constant $\delta$, we give explicit constructions that approach these limits, providing robustness up to just under half of the bits (for binary) or just under $(1-1/q)$ (for $q$-ary) of the symbols changed.

To make our results concrete, we analyze a recent state-of-the-art PRC-based watermarking scheme for images by Gunn, Zhao, and Song~\cite{gunn2025undetectable}. We show that a simple crop-and-resize operation, one that visually preserves the image but changes about half of its latent bits, is sufficient to consistently prevent belief-propagation decoding from recovering the codeword. This suggests that the theoretical limit we establish is relevant in at least one practical watermarking pipeline.

This attack also highlights a contrast between modalities: in text, watermark removal typically requires extensive edits that alter the content, while for images, even a single benign transformation can remove the watermark without changing the image in an obvious way. This makes robust watermarking in images particularly fragile, since the watermark can disappear without a visible cost.

Our findings position the impossibility results of WiTS in a precise, quantitative framework for PRC-style cryptographic watermarks: WiTS demonstrates that all robust watermarks are eventually breakable in a powerful black-box oracle model, while our work identifies a concrete numerical threshold at which robustness becomes unconditionally impossible for schemes that can be modeled as zero-bit tamper-detection codes under independent symbol corruption.

Looking forward, our results suggest that any significant increase in robustness within this PRC-based, cryptographic framework will require fundamentally new ideas, potentially leveraging semantic or structural features of content rather than cryptographic pseudorandomness alone.

\subsection{Our Contributions}

Our contributions are as follows:

\begin{enumerate}[leftmargin=*]
\item \textbf{A simplified abstraction for PRC-style watermarking.}
We introduce and formalize \emph{zero-bit tamper-detection codes}, a secret-key primitive that we use as a minimal model for cryptographic watermarking based on pseudorandom codes. These codes do not encode explicit messages but instead focus on two properties: \emph{soundness} and \emph{tamper detection}. They can be viewed as the zero-bit specialization of pseudorandom error-correcting codes equipped with an explicit tamper-detection predicate (and without pseudorandom codewords). We show how any watermarking scheme that is sound and robust for a given family of tampering functions induces a zero-bit tamper-detection code with the same guarantees, even when the watermark detector depends on the prompt. In this sense, zero-bit tamper-detection codes capture the core coding-theoretic requirements of PRC-based watermarking for generative models.



\item \textbf{Tight unconditional limits for watermarking schemes.}
Within this abstraction we establish a sharp limit on robustness to independent symbol tampering. For an alphabet of size $q$ there is a critical corruption rate
\[
\alpha^* = 1 - \frac{1}{q}
\]
such that no zero-bit tamper-detection code can simultaneously satisfy soundness and tamper detection for independent tampering at rates exceeding $\alpha^*$ by any fixed constant slack. In particular, in the binary case ($q = 2$) no such code can remain both sound and robust once more than half of the encoded bits are modified. Conversely, we give a simple information-theoretic construction that achieves soundness and tamper detection for all corruption rates strictly below $\alpha^*$, up to an arbitrarily small constant slack. 

We further show that this barrier persists even if soundness is weakened to \emph{constant-bounded soundness}. For any constant $c \in (0,1]$, once an adversary can corrupt more than an $(\alpha^* + \varepsilon)$ fraction of symbols (for any fixed $\varepsilon > 0$), any code that is sound on all but a $c$ fraction of random strings must incur tamper-detection error close to $1 - c$ at that corruption level. In particular, even very weak schemes that behave like a coin flip on random inputs (for example, $c = 1/2$) cannot robustly detect tampering beyond this threshold.

\item \textbf{Implications for PRC-based watermarking.}
By relating zero-bit tamper-detection codes to pseudorandom error-correcting codes, and using the reduction from watermarking schemes to tamper-detection codes, we obtain direct consequences for PRC-based watermarking. In particular, for the independent symbol tampering model considered by Christ and Gunn~\cite{C:ChrGun24}, our lower bound shows that their robustness guarantees are  optimal: no cryptographic watermark (binary or $q$-ary, secret-key or public-key) whose robustness can be expressed through such a coding layer can tolerate a larger fraction of symbol errors in this model.

\item \textbf{A concrete attack illustrating the limit.}
Finally, we show that the coding-theoretic limit is numerically relevant for at least one state-of-the-art watermarking pipeline. We analyze the PRC-based image watermark for Stable Diffusion proposed by Gunn, Zhao, and Song~\cite{gunn2025undetectable}, and we study the effect of a simple crop-and-resize transformation. This transformation preserves the visual appearance of the image yet induces latent sign errors concentrated near fifty percent and consistently causes belief-propagation decoding to fail, thereby erasing the watermark.

Earlier work had already established that watermarks can, in principle, be removed. WiTS shows this using oracle access and iterative regeneration~\cite{zhang2023watermarks,zhang2024watermarks}, while UnMarker demonstrates removal in a fully black-box setting using optimization-based spectral attacks that slightly degrade the image~\cite{unmarker2025}. Our case study differs in that it reaches essentially the same corruption regime using only a single, routine image operation, with no oracle access and no optimization. In this sense, it provides a concrete example where the unconditional robustness threshold arises from an ordinary edit rather than a highly tailored attack.
\end{enumerate}

\subsection{Related work}
Watermarking refers to the process of embedding a signal in generated content such as images~\cite{boland_95,cox_et_al_07,Hayes_et_al_17,Ruanaidh_1997,Zhu_18}, text~\cite{text_1,text_2}, audio~\cite{arnold_00,Boney_96}, or video that is imperceptible to humans but algorithmically detectable. Its primary use is to enable attribution or provenance without significantly degrading content quality. The study of watermarking spans decades and includes both practical systems and formal analyses, with results on constructing watermarking schemes as well as breaking or proving impossibility. With the rise of modern generative models, watermarking has become a sharper test of the trade-offs between imperceptibility, robustness, and efficiency.

\ifconf
\else
\paragraph{Constructions.}
\fi
Watermarking methods are commonly classified based on when the watermark is applied: either after generation (post-processing) or during the generation process itself (in-processing). Post-processing schemes embed the watermark into the output after it has been generated. These approaches often degrade output quality~\cite{cox_et_al_07,Wan_22,An_24}. In-processing schemes, by contrast, embed the watermark during content generation. They were first introduced in image generators~\cite{Yu_21}, and later adapted to large language models (LLMs)~\cite{Kirchenbauer_23}. Image watermarking techniques include modifying weights through pretraining or fine-tuning~\cite{Fei_22,Fernandez_23,Zeng_23,Zhao_23c}, or intervening in the sampling trajectory of diffusion models~\cite{wen_2023}. In the LLM setting, most methods operate at inference time by biasing token selection in a way that encodes a hidden message~\cite{Kirchenbauer_23,Christ_23}. Some schemes further ensure that the overall distribution over outputs remains statistically indistinguishable from unmarked outputs~\cite{kuditipudi_24}, preserving quality. 

A more recent line of work leverages cryptographic and coding-theoretic tools to design undetectable watermarking schemes. One direction~\cite{C:ChrGun24} proceeds via pseudorandom error-correcting codes (PRCs) as the watermark signals. These codes can be constructed using low-density parity-check (LDPC) codes~\cite{C:ChrGun24,STOC:AACDG25,ghentiyala_24} and allow for robust detection under edits. PRC-based watermarking was recently applied to images~\cite{gunn2025undetectable}. The main limitation~\cite{gunn2025undetectable} is that although the underlying PRC is robust to modifications on a bounded number of bits, the resulting PRC-based watermarking is robust only to independent changes, limiting its applicability in real-world scenarios. In parallel, a recent public watermarking scheme uses digital signatures, allowing third parties to verify the watermark and ensuring robustness as long as a small window of the generated content remains unmodified~\cite{garg23}.

\ifconf
\else
\paragraph{Attacks and Impossibilities.}
\fi
Alongside construction efforts, recent work has explored the limitations of existing watermarking schemes. In text, the most prominent attack is paraphrasing~\cite{paraphrasing_attack_1,paraphrasing_attack_2,zhang2023watermarks,zhang2024watermarks}, rewriting the watermarked content while preserving its meaning, using language models, translation tools, or manual edits. These attacks often succeed partially, but at the cost of degraded output quality. In the vision domain, a range of attacks aim to remove watermarks by adding noise to the image or latent representation, or by using optimization techniques to remove the watermark~\cite{lukas_24,saberi_24,Zhao_24b}. Many of these attacks have limitations.
For example, Zhao et al.~\cite{Zhao_24b} focus on post-hoc watermarking methods, while Saberi et al.~\cite{saberi_24} require either white-box access to the detector or access to many watermarked and unwatermarked examples. 

A more general impossibility result by Zhang et al.~\cite{zhang2024watermarks,zhang2023watermarks} shows that constructing robust watermarking schemes is impossible under certain conditions. They model generated outputs as nodes in a graph and prove that if an adversary has access to two oracles (a quality oracle that tells the adversary whether the output is still valid, and a perturbation oracle that allows small edits preserving semantics) then the adversary can walk through the graph until the watermark is erased without degrading content quality. This result does not specify an exact condition for when removal occurs, only that eventually such a walk succeeds. 

In contrast, our result makes no oracle assumptions and provides a concrete lower bound in an independent symbol corruption model. We show that if more than half of the encoded bits are effectively randomized, then no watermark can remain both sound and robust. Prior work demonstrates that robust watermarking cannot be sustained indefinitely. Our result identifies a precise corruption threshold in this setting beyond which robustness is unconditionally impossible.


\section{Preliminaries}\label{sec:prel}
\subsection{Notation}\label{sec:notation}
Small letters (such as $x$) denote individual objects or values; calligraphic letters (such as $\mathcal{X}$) denote sets; sans-serif letters (such as $\mathsf{A}$) denote algorithms.
For a string $x \in \{0,1\}^*$, $|x|$ denotes its length. For a set $\mathcal{X}$, $|\mathcal{X}|$ is its cardinality. If $x$ is chosen uniformly at random from $\mathcal{X}$, we write $x \getsr \mathcal{X}$.
The Hamming distance $\dist(x, x')$ between two strings $x$ and $x'$ (over an alphabet $\alphabet$) is the number of positions where they differ.

For a deterministic algorithm $\mathsf{A}$, $y = \mathsf{A}(x)$ means that $y$ is the output of $\mathsf{A}$ on input $x$. For a randomized algorithm $\mathsf{A}$, we write $y \getsr \mathsf{A}(x)$ to indicate $y$ is sampled from the output distribution of $\mathsf{A}$ on input $x$. Alternatively, $y = \mathsf{A}(x; r)$ denotes the output when $\mathsf{A}$ runs on $x$ with randomness $r$. An algorithm $\mathsf{A}$ is called \emph{probabilistic polynomial-time} (PPT) if it is randomized and always halts in time polynomial in $|x|$.

We use $\mathsf{negl}(\secpar)$ for an arbitrary negligible function of the security parameter $\secpar \in \mathbb{N}$, i.e., for all $c > 0$, $\mathsf{negl}(\secpar) = o(\secpar^{-c})$. We use $\poly(\secpar)$ to denote a polynomial function of $\secpar$. Unless noted otherwise, all algorithms receive the security parameter $1^\secpar$ as input.

\subsection{Watermarking for Generative Models}
\label{sec:watermarking}
We recall the formalization of watermarking for generative models, following~\cite{zhang2023watermarks,zhang2024watermarks}.

A generative model is any (possibly randomized) algorithm that, given a prompt $x$ (such as a question or a text description), produces an output $y$ (such as text or an image).

\begin{defn}[Generative models]
\label{def:gen-models}
A conditional generative model $\LLM: \cX \rightarrow \cY$ is a probabilistic polynomial-time (PPT) algorithm that, given a prompt $x \in \cX$, outputs $y \in \cY$. Here, $\cX$ is the \emph{prompt space}, and $\cY$ is the \emph{output space}. We write $y \getsr \LLM(x)$ to denote sampling a response from the model on prompt $x$.
\end{defn}

A watermarking scheme is a systematic way to mark the outputs of a generative model, so that later a verifier with the appropriate key can detect whether a given output is watermarked. We focus on \emph{secret-key} watermarking: both embedding and detection require a secret key.

A watermarking scheme $\Pi = (\Watermark, \Detect)$ for a family of generative models $\cModel = \{\LLM\}$ consists of two efficient algorithms:
\begin{description}
    \item[$\Watermark(1^\lambda, \LLM)$:] On input a security parameter $1^\lambda$ and a model $\LLM$, the algorithm outputs a secret key $\kappa \in \cK$ and a watermarked model $\LLM_\kappa: \cX \rightarrow \cY$.
    \item[$\Detect(\kappa, x, y)$:] On input a secret key $\kappa$, a prompt $x \in \cX$, and an output $y \in \cY$, this \blue{deterministic} algorithm returns either $\mathsf{true}$ or $\mathsf{false}$, indicating whether $y$ is watermarked (in response to $x$).
\end{description}

A robust watermarking scheme should satisfy three key properties.
First, \textbf{correctness}: watermarked outputs should always be recognized as such.

\begin{defn}[Correctness of watermarking]
\label{def:corr-models}
Let $\cModel = \{\LLM:\cX \rightarrow \cY\}$ be a class of generative models. We say that $\Pi$ satisfies \emph{correctness} if for every $\LLM \in \cModel$ and for every $x \in \cX$, we have
\ifconf
\else
\[
\Prob\left[\Detect(\kappa, x, y) = \true: (\kappa, \LLM_\kappa) \getsr \Watermark(1^\lambda,\allowbreak \LLM),\ y \getsr \LLM_\kappa(x)\right] = 1.
\]
\fi
\end{defn}

Second, \textbf{soundness}: non-watermarked outputs (such as human-written or independently generated model outputs) should almost never be mistakenly detected as watermarked. 

\begin{defn}[Soundness of watermarking]
\label{def:sound-watermark}
Let $\cModel = \{\LLM:\cX \rightarrow \cY\}$ be a class of generative models. We say that $\Pi$ satisfies \emph{soundness} if for every $\LLM \in \cModel$, every $x \in \cX$, and every \blue{fixed} $y \in \cY$, 
\ifconf
\else
$\Prob\left[\Detect(\kappa, x, y) = \true: (\kappa, \LLM_\kappa) \getsr \Watermark(1^\lambda, \LLM)\right] \leq \negl(\lambda)$.
\fi
\end{defn}

Third, \textbf{robustness}: the watermark should survive small modifications to the output.
Robustness requires that a watermark remain detectable even if an adversary modifies a small part of the output. To capture this, we consider \blue{length-preserving \emph{tampering functions} $f : \Sigma^* \rightarrow \Sigma^*$, where $\alpha \in (0,1)$ specifies the target fraction of edits relative to the input length}.

We study two models of adversarial tampering:
(1) In the \emph{arbitrary tampering} model, the adversary may choose any positions and values for up to \blue{an $\alpha$ fraction of changes}, allowing fully coordinated edits.
(2) In the \emph{independent tampering} model, \blue{the adversary changes each symbol independently, and we require that the total number of changes is at most an $\alpha$ fraction of the input length with overwhelming probability}.

Let $\hat\cF_\alpha$ denote the set of (possibly randomized) functions $f$ as above (arbitrary), and let $\hat\cF^{\mathrm{ind}}_\alpha \subseteq \hat\cF_\alpha$ denote those which act independently on each symbol.

\begin{defn}[$\hat\cF$-Robustness of watermarking]
\label{def:robust-watermark}
Let $\cModel = \{\LLM:\blue{\cX \rightarrow \Sigma^*}\}$ be a class of generative models.
A watermarking scheme $\Pi$ is \emph{$\hat\cF$-robust} if for every $\LLM \in \cModel$, every $x \in \cX$, and every $f \in \hat\cF_\alpha$, we have
\ifconf
\else
\[
\Prob\left[\Detect(\kappa, x, \tilde y) = \false \land \tilde y \ne y:
\begin{array}{c}
(\kappa, \LLM_\kappa)\getsr\Watermark(1^\lambda, \LLM) \\
y \getsr \LLM_\kappa(x);\; \tilde y \getsr f(y)
\end{array}
\right] \leq \negl(\lambda).
\]
\fi
\end{defn}

In our results, we primarily consider small constants $\alpha > 0$. When we refer to $\hat\cF^{\mathrm{ind}}_\alpha$, the impossibility results are strongest, in the sense that we show robustness is already impossible against this restricted class of independent, uncoordinated edits. This immediately implies impossibility against the larger class of arbitrary tampering functions.

We always require $|f(y)| = |y|$, meaning all tampering functions are length-preserving. This restriction is standard in coding theory and only strengthens our impossibility results: if robust watermarking is impossible even when the adversary must preserve output length, it is certainly impossible when more general modifications, such as insertions or deletions, are permitted.

Our model further allows the watermarking algorithm to access the generative model's internals, and allows the detector to depend on both the prompt and output. Impossibility under these permissive conditions immediately implies impossibility in any more restricted setting.

\section{Zero-bit Tamper-detection Codes}\label{sec:tdc}

Our lower bounds will be phrased in terms of a simple coding primitive that captures the behavior of watermark detectors. Informally, we imagine that for each secret key there is a single special codeword in $\alphabet^\cdwlen$. Given a string $y \in \alphabet^\cdwlen$, the decoder must decide whether
\begin{enumerate}
    \item $y$ is the honest codeword produced under the key,
    \item $y$ is unrelated to any codeword for that key, or
    \item $y$ should be treated as a tampered version of the honest codeword.
\end{enumerate}
There is no explicit message to recover; the output is only this three-way classification. We refer to such objects as \emph{zero-bit tamper-detection codes}.

We work in a secret-key setting: both the encoder and the decoder have access to the same key $\sk$, which is hidden from the adversary. This is a favorable model for the designer of the code. Any public-key or keyless detector can be viewed as a special case in which the key is fixed and revealed in advance. Thus, if we show that no scheme exists even in this more permissive secret-key model, the impossibility automatically applies to all more restricted variants.

It is helpful to compare this notion with a familiar primitive. Suppose we have a secure message authentication code with tagging algorithm $\mathsf{Tag}_\sk(\cdot)$. Fix a message $m = 0$, and for a given key $\sk$ let $\cdw = \mathsf{Tag}_\sk(0)$. We can define an encoder that always outputs $\cdw$, and a decoder that declares $\valid$ on input $\cdw$ and $\tampered$ on every other input. This construction detects any change to $\cdw$ perfectly. In this sense, detecting modifications of a single designated string is easy. The difficulty arises once we also require the decoder to behave well on \emph{all} other strings that might be presented to it. The definitions below make this requirement precise.

Formally, a zero-bit tamper-detection code is specified by a triple of polynomial-time algorithms
\[
\Gamma = (\KGen, \Encode, \Decode)
\]
over an alphabet $\alphabet$ and a codeword space $\alphabet^\cdwlen$:
\begin{description}
\item[$\KGen(1^\secpar)$:] On input the security parameter $\secpar$, the randomized key-generation algorithm outputs a secret key $\sk$.
\item[$\Encode(\sk)$:] On input $\sk$, the (possibly randomized) encoding algorithm outputs a codeword $\cdw \in \alphabet^\cdwlen$.
\item[$\Decode(\sk,\cdw)$:] On input the secret key $\sk$ and a string $\cdw \in \alphabet^\cdwlen$, the deterministic decoding algorithm outputs a symbol in $\{\valid,\invalid,\tampered\}$.
\end{description}

The first requirement is that honestly generated codewords are always accepted.

\begin{defn}[Correctness of zero-bit tamper-detection codes]\label{def:corr-no-msg}
A zero-bit tamper-detection code $\Gamma$ satisfies \emph{correctness} if, for every $\secpar \in \NN$ and every secret key $\sk \in \supp(\KGen(1^{\secpar}))$, we have
\[
\Prob\bigl[\Decode(\sk, \Encode(\sk)) = \valid\bigr] = 1,
\]
where the probability is taken over the internal randomness of $\Encode$ (if any). In particular, if $\Encode$ is deterministic then $\Decode(\sk,\Encode(\sk)) = \valid$ holds for all such $\sk$.
\end{defn}

To capture security, we add two further conditions. The first is \emph{soundness}. Here we fix an arbitrary string $\hat\cdw \in \alphabet^\cdwlen$ that is chosen independently of the key. Soundness requires that, over the random choice of $\sk$, the decoder almost never classifies $\hat\cdw$ as anything other than $\invalid$. Intuitively, for a random key the set of strings that the decoder considers valid or tampered should occupy only a small fraction of the ambient space.

\begin{defn}[Soundness of zero-bit tamper-detection codes]\label{def:sound-no-msg}
We say that a zero-bit tamper-detection code $\Gamma$ satisfies \emph{soundness} if for every fixed string $\hat\cdw \in \alphabet^\cdwlen$,
\[
\Prob\left[\Decode(\sk,\hat\cdw) \ne \invalid : \sk \getsr \KGen(1^\secpar)\right] \le \negl(\secpar).
\]
\end{defn}

For our lower bounds we will also use a relaxed form of this requirement. Rather than insisting that the misclassification probability be negligible, we allow it to be bounded by a fixed constant \(c\). From the point of view of applications, such a guarantee is very weak. For example, if \(c = 1/2\), then in the worst case the decoder can do no better than a fair coin flip on a fixed string. The relaxation is useful for the structure of our result: the class of schemes that satisfy negligible soundness is contained inside the larger class that satisfy constant-bounded soundness for some \(c < 1\). If we can rule out robust tamper detection for the larger class, then we have automatically ruled it out for the smaller class as well.

\begin{defn}[Constant-bounded soundness of zero-bit tamper-detection codes]\label{def:coin-sound}
Let $c \in (0,1]$. We say that a zero-bit tamper-detection code $\Gamma$ satisfies \emph{constant-bounded soundness} (with parameter $c$) if for every fixed string $\hat\cdw \in \alphabet^\cdwlen$,
\[
\Prob\left[\Decode(\sk,\hat\cdw) \ne \invalid : \sk \getsr \KGen(1^\secpar)\right] \le c.
\]
If $c = \tfrac{1}{2}$, we say that $\Gamma$ satisfies \emph{coin-flipping soundness}, in the sense that checking the validity of a worst-case string is no better than flipping a fair coin.
\end{defn}

The second security condition concerns the effect of tampering on an honestly generated codeword. We model tampering by a family $\cF = \{f:\alphabet^\cdwlen\rightarrow\alphabet^\cdwlen\}$ of (possibly randomized) functions. For a given key, the encoder produces $\cdw \getsr \Encode(\sk)$ and the adversary applies some $f \in \cF$ to obtain $\tilde\cdw$. Tamper detection requires that, except with negligible probability, the decoder outputs $\tampered$ whenever $\tilde\cdw$ differs from $\cdw$.

\begin{defn}[$\cF$-tamper detection of zero-bit tamper-detection codes]\label{def:tamper-no-msg}
We say a zero-bit tamper-detection code $\Gamma$ satisfies \emph{$\cF$-tamper detection} if for every $f\in\cF$,
\ifconf
\else
    \[
    \Prob\left[\Decode(\sk,\tilde\cdw) \ne \tampered \wedge \tilde\cdw \ne \cdw :
    \sk\getsr\KGen(1^\secpar);\; \cdw \getsr \Encode(\sk);\; \tilde\cdw \getsr f(\cdw)\right] \le \negl(\secpar).
    \]
\fi
\end{defn}

Correctness, soundness, and $\cF$-tamper detection together describe the behavior we will require from the coding layer underlying cryptographic watermarking.

\subsection{Relation to PRCs}
Christ and Gunn~\cite{C:ChrGun24} formalize cryptographic watermarking in terms of pseudorandom error-correcting codes (PRCs). A PRC is a keyed coding scheme that both corrects a bounded amount of noise and hides its codewords from anyone who does not know the key. Since our lower bounds are meant to apply to PRC-based watermarking, it is natural to ask why we did not state them directly in that language.

The main reason is that we work in the information theoretic setting to prove the negative result.\footnote{Despite this, the negative result can be naturally extended even to computational security.} It depends only on how many symbols an adversary can change and on how the decoder responds to corrupted strings. The pseudorandomness of the codeword, which is central to undetectability, plays no role in this counting argument. Zero-bit tamper-detection codes isolate exactly the pieces we need: a keyed encoder, a decoder that distinguishes $\valid$, $\invalid$, and $\tampered$, and quantitative guarantees for soundness and tamper detection with respect to a family $\cF$ of tampering functions. We do not constrain the distribution of codewords beyond what is needed for these properties.

Viewed this way, zero-bit tamper-detection codes form a weaker abstraction than PRCs. A zero-bit PRC must in addition make its codewords pseudorandom and is usually required to have negligible decoding error for every message. Any impossibility theorem that already rules out robust zero-bit tamper detection under these minimal assumptions therefore applies a fortiori to zero-bit PRCs that satisfy stronger guarantees. Working with the weaker interface keeps the statement of the lower bound focused on the two properties that matter for our application (soundness and tamper detection) while still covering the PRC constructions used in cryptographic watermarking.

At the same time, in the zero-bit setting the coding content of the two formalisms is essentially the same once we impose standard soundness. There is only a single message $\msg$ to recover. Given a zero-bit tamper-detection code $\Gamma = (\KGen,\Encode,\Decode)$ that satisfies soundness in the sense of \Cref{def:sound-no-msg} and $\cF$-tamper detection in the sense of \Cref{def:tamper-no-msg}, we can view it as a (non-pseudorandom) zero-bit error-correcting code for the fixed message $\msg$ by using $\Encode$ as the encoder and defining the decoder to output $\msg$ on any input that is not classified as $\invalid$. Conversely, for the purposes of our lower bound, any zero-bit PRC with negligible decoding error in a given tampering model \blue{induces a degenerate tamper-detection code by outputting $\tampered$ on successful decodings and $\invalid$ otherwise. This induced code need not satisfy correctness}, since it may never output $\valid$, but it preserves the soundness and tamper-detection guarantees relevant to the lower bound up to the usual negligible losses.

Thus, if we strip away pseudorandomness, zero-bit tamper-detection codes and zero-bit PRCs express the same coding-theoretic limitations in the regime we study. Our lower bound for zero-bit tamper-detection codes immediately carries over to zero-bit PRCs whose decoder is required to be robust against the same independent symbol tampering model. Adding pseudorandomness on top can improve undetectability of the watermark, but it cannot move the underlying coding limit.

Finally, the limit we obtain is tight for this abstraction. In Appendix~\ref{app:construction} we describe a simple information-theoretic construction of a zero-bit tamper-detection code that tolerates any independent tampering rate strictly below the threshold identified by our bound. This matching construction shows that the impossibility is not an artifact of the definition: it captures the exact corruption level at which robust decoding, and hence robust PRC-based watermarking, becomes information-theoretically impossible.



\subsection{Impossibility of High Tampering Rates}\label{sec:limit-no-msg}

A central question for zero-bit tamper-detection codes is how robust they can be to tampering. If an adversary is allowed to modify a fraction of the symbols in a codeword, can the decoder still reliably detect that tampering has occurred while maintaining soundness on arbitrary strings? In this subsection we show that these two requirements come into fundamental conflict once the tampering rate is high enough.

To build some intuition, consider an extreme case in which the tampering family contains all constant functions: for every fixed string $\hat\cdw \in \alphabet^\cdwlen$, there is a function $f_{\hat\cdw}$ that ignores its input and always outputs $\hat\cdw$. Soundness says that, over the choice of the key, the decoder should almost always label $\hat\cdw$ as $\invalid$. On the other hand, if $\hat\cdw$ can arise as $f_{\hat\cdw}(\cdw)$ for an honestly generated codeword $\cdw$, tamper detection would like the decoder to label $\hat\cdw$ as $\tampered$. Since the decoder can assign only one label to each string, it cannot satisfy both requirements for the same $\hat\cdw$. This already suggests that strong tamper detection is hard to reconcile with soundness once we allow very powerful tampering.

Our main result shows that a similar tension appears even for a much weaker, randomized tampering model in which each symbol is modified independently. Let $\alphabet$ be an alphabet of size $q>1$, and let $n$ denote the codeword length. We focus on the following independent tampering rule: for each input $\cdw \in \alphabet^n$, the adversary produces $\tilde\cdw$ by choosing each symbol of $\tilde\cdw$ independently and uniformly from $\alphabet$, conditional on $\cdw$. For each position, the new symbol equals the original one with probability $1/q$ and differs from it with probability:
\[
1 - \frac{1}{q}.
\]
\blue{Call this randomized tampering function $f_{\mathrm{unif}}$.} Thus the expected fraction of changed positions is $1 - 1/q$. For any fixed $\delta \in (0,1)$, a Chernoff bound gives
\begin{align}
\Prob\!\left[\dist(\cdw,f_{\mathrm{unif}}(\cdw)) > \Bigl(1 - \frac{1}{q}\Bigr)(1+\delta)\,n\right] \nonumber \\
\le \exp\!\left(- \frac{\delta^2}{3} n \Bigl(1-\frac{1}{q}\Bigr)\right)\label{eq:validity}
\end{align}
for every fixed $\cdw \in \alphabet^n$. 
For the theorem below, it suffices that \blue{$\cF^{\mathrm{ind}}_{\alpha n}$ be any tampering family that contains $f_{\mathrm{unif}}$, where $\alpha = \Bigl(1 - \frac{1}{q}\Bigr)(1+\delta)$; according to~\Cref{eq:validity}, $f_{\mathrm{unif}} \in \cF^{\mathrm{ind}}_{\alpha n}$ holds with overwhelming probability when considering sufficiently long codewords (i.e., large values of $n$).}

\blue{For simplicity, we state the main impossibility theorem below under the assumption that $f_{\mathrm{unif}} \in \cF^{\mathrm{ind}}_{\alpha n}$. Thus, it holds except with the probability defined in~\Cref{eq:validity}.}

\begin{thm}\label{thm:impossible}
    Let $\Gamma$ be any zero-bit tamper-detection code over $\alphabet$ with codeword length $n \geq 1$ and $|\alphabet| = q > 1$. Suppose $\Gamma$ satisfies \blue{$\cF^{\mathrm{ind}}_{\alpha n}$-tamper detection for some tampering family $\cF^{\mathrm{ind}}_{\alpha n}$ that contains $f_{\mathrm{unif}}$}, where
    \[
      \alpha = \Bigl(1 - \frac{1}{q}\Bigr)(1+\delta)
    \]
    with $\delta \in (0,1)$. Then $\Gamma$ cannot satisfy soundness.
\end{thm}

\begin{proof}
    Suppose, for contradiction, that such a code $\Gamma$ exists.

\blue{Let $\sk \getsr \KGen(1^\secpar)$ and let $\cdw = \Encode(\sk)$ be an honestly generated codeword. Let $\cdw' \getsr f_{\mathrm{unif}}(\cdw)$.}

\blue{By definition of $f_{\mathrm{unif}} \in \cF^{\mathrm{ind}}_{\alpha n}$, the random variable $\cdw'$ is distributed exactly as a uniformly random string in $\alphabet^n$.}

By the soundness property, for every fixed $\hat\cdw \in \alphabet^n$,
\[
\Prob\left[\Decode(\sk, \hat\cdw) \neq \invalid \right] \leq \negl(\secpar),
\]
where the probability is over $\sk \getsr \KGen(1^\secpar)$. In particular, this bound also holds when $\hat\cdw$ is chosen uniformly at random, so for a randomly chosen $\hat\cdw \getsr \alphabet^n$ we still have $\Prob\left[\Decode(\sk, \hat\cdw) \neq \invalid \right] \leq \negl(\secpar)$.

\blue{Moreover, since $f_{\mathrm{unif}} \in \cF^{\mathrm{ind}}_{\alpha n}$, tamper detection gives}
\[
\blue{
\Prob\left[\Decode(\sk, \cdw') \neq \tampered \wedge \cdw' \neq \cdw \right] \leq \negl(\secpar),
}
\]
where the probability is over $\sk \getsr \KGen(1^\secpar)$, $\cdw \getsr \Encode(\sk)$, and the randomness of $f_{\mathrm{unif}}$. \blue{Since $\cdw'$ is uniform over $\alphabet^n$, we have $\Prob[\cdw' = \cdw] = q^{-n}$. Therefore,}
\[
\blue{
\Prob\left[\Decode(\sk, \cdw') = \tampered \right] \geq 1 - \negl(\secpar) - q^{-n}.
}
\]

Therefore, under the joint distribution of $(\sk,\cdw')$ with $\sk \getsr \KGen(1^\secpar)$ and $\cdw' \getsr \alphabet^n$, we have both
\[
\blue{\begin{aligned}
\Prob\left[\Decode(\sk, \cdw') = \invalid\right] &\geq 1 - \negl(\secpar),\\
\Prob\left[\Decode(\sk, \cdw') = \tampered\right] &\geq 1 - \negl(\secpar) - q^{-n},
\end{aligned}}
\]
\blue{These events are disjoint. Since $q^{-n} \leq 1/2$ for every $q>1$ and $n \geq 1$, these two lower bounds sum to more than $1$ for all sufficiently large $\secpar$, which is impossible.}

We conclude that no zero-bit tamper-detection code can be sound and $\cF^{\mathrm{ind}}_{\alpha n}$-tamper-detecting at tampering rates $\alpha \geq (1-1/q)(1+\delta)$.

\end{proof}

In words, once we ask for tamper detection against independent symbol modifications that, with high probability, corrupt slightly more than a $(1-1/q)$ fraction of the positions, there is no way to keep the decoder sound on arbitrary fixed strings. The value
\[
\alpha^* = 1 - \frac{1}{q}
\]
is the critical corruption rate in this model: below it, robust tamper detection is achievable (as we show with a matching construction), while above it, soundness and tamper detection cannot coexist.

In the binary case the statement takes a particularly clean form.

\begin{coro}\label{coroll:impossible-bits}
    Let $\Gamma$ be any zero-bit tamper-detection code with alphabet $\bin$ and codeword space $\bin^n$ for $n \geq 1$. If $\Gamma$ satisfies \blue{$\cF^{\mathrm{ind}}_{\alpha n}$-tamper detection for some tampering family $\cF^{\mathrm{ind}}_{\alpha n}$ that contains $f_{\mathrm{unif}}$}, and for
    \[
       \alpha = \frac{1+\delta}{2}
    \]
    for some $\delta \in (0,1)$, then $\Gamma$ cannot also satisfy soundness.
\end{coro}

Notice that these impossibility results do not rely on correctness: nowhere in the proof do we use the requirement that an honestly generated codeword must be accepted as $\valid$. Even a degenerate scheme that never outputs $\valid$ would trivially satisfy soundness (since all random strings are labeled $\invalid$) and might still claim to detect tampering. The theorem shows that, at high tampering rates, such a scheme also cannot satisfy the $\cF^{\mathrm{ind}}_{\alpha n}$-tamper detection property. This is why correctness is not assumed in our lower bounds.

\paragraph{Extension to weak soundness and variable length}
The same phenomenon persists even if we weaken soundness substantially. Recall the definition of constant-bounded soundness (\Cref{def:coin-sound}): for a parameter $c \in (0,1]$, the decoder is allowed to misclassify a fixed string as anything other than $\invalid$ with probability up to $c$ over the choice of the key. For the independent tampering model above and corruption rates
\[
\alpha = \Bigl(1 - \frac{1}{q}\Bigr)(1+\delta),
\]
one can show that if a code has constant-bounded soundness with parameter $c$, then for large $n$ its tamper-detection error under such tampering must be at least
\[
(1 - c)\Bigl(1 - \exp\bigl(-\tfrac{\delta^2}{3}(1-1/q)\,n\bigr)\Bigr),
\]
which tends to $1-c$ as $n$ grows. In particular, even a ``coin-flip'' decoder with $c = 1/2$ cannot hope to detect high-rate independent tampering with small error; beyond the threshold $\alpha^*$ it must fail with probability close to $1/2$ on tampered codewords.

Finally, the impossibility does not depend on fixing the blocklength in advance. Allowing the code to use different lengths $n$ on different inputs does not help: for each length $n$ we can define the corresponding independent tampering rule on $\alphabet^n$, and the same argument goes through. In this sense the barrier identified by \Cref{thm:impossible} is intrinsic to the combination of soundness and strong tamper detection, rather than an artifact of a particular choice of codeword length.

\ifconf
\else
\section{Connecting Watermarking and Tamper-detection Codes}\label{sec:water-tdc}
\fi
The fundamental limits we have established for zero-bit tamper-detection codes also govern the robustness of watermarking schemes for generative models. To make this connection precise, we show how any watermarking scheme that is sound and robust against tampering can be used to build a secret-key code with exactly the same security guarantees.

The intuition is straightforward: if a generative model can reliably embed a watermark that survives tampering, then by fixing a prompt and interpreting the model's outputs as codewords, we obtain a code that is robust to exactly the same set of manipulations. The code's decoder simply runs the watermark detector: if it finds a watermark, it signals tampering; otherwise, it outputs invalid.

For the formal theorem below, we restrict to models whose outputs have \blue{fixed length $n$}. In that case, the resulting code fits the \blue{fixed-length zero-bit tamper-detection interface} from \Cref{sec:tdc}.

\begin{thm}\label{thm:watermark-implies-tdc}
    Assume there exists a watermarking scheme $\Pi$ for a class of generative models \blue{$\cModel = \{\LLM:\cX\rightarrow\Sigma^n\}$} satisfying soundness and $\hat\cF$-robustness (for some family $\hat\cF$).\footnote{For the sake of clarity, we assume the prompt and output space of $\LLM$ is defined over the same alphabet. The result generalizes to different alphabets as well.} \blue{Fix any model $\LLM \in \cModel$ and any prompt $\overline{x} \in \cX$.} Then, there exists a zero-bit tamper-detection code $\Gamma$ with alphabet space $\alphabet$ and \blue{codeword space $\alphabet^n$} satisfying soundness and $\hat\cF$-tamper detection.
\end{thm}
\begin{proof}
    \blue{We fix $\LLM \in \cModel$ and $\overline{x} \in \cX$ as in the theorem statement, and describe the resulting zero-bit tamper-detection code $\Gamma$.}
    \begin{description}
        \item[Key generation $\Gamma.\KGen(\secparam)$:] The key-generation algorithm outputs $\sk = (\kappa,\LLM_\kappa)$ where \blue{$(\kappa,\LLM_\kappa)\getsr \Pi.\Watermark(\secparam,\LLM)$}.
        \item[Encoding $\Gamma.\Encode(\sk)$:] The encoding algorithm outputs $\cdw = y \in\blue{\Sigma^{n}}$ where $y \getsr \LLM_\kappa(\overline{x})$.
        \item[Decoding $\Gamma.\Decode(\sk,\cdw)$:] The decoding algorithm lets $\cdw = y$ and runs $\Pi.\Detect(\kappa,\overline{x},y)$. If the result is $\false$, it outputs $\invalid$. If the result is $\true$, it outputs $\tampered$.
    \end{description}
    Note that the decoder never outputs $\valid$, and thus $\Gamma$ does not satisfy correctness (which is not required here).

    Let us first prove the soundness property. By contradiction, assume $\Gamma$ is not sound. Then, there exists some $\hat\cdw \in\blue{\Sigma^n}$ such that
    \ifconf
\else
    \[
    \Prob\left[\Gamma.\Decode(\sk,\hat\cdw) \ne \invalid\right] = \Prob\left[\Pi.\Detect(\kappa,\overline{x},\hat\cdw) = \true\right] \ge 1/\poly(\secpar),
    \]
    \fi
    where the probability is over the choice of the secret key. This violates the soundness of the watermarking scheme.

    It remains to prove tamper detection. By contradiction, assume $\Gamma$ does not satisfy $\hat\cF$-tamper detection. Then, there exists some \blue{(possibly randomized)} function $\hat f\in\hat\cF$ such that
    \ifconf
\else
        \[
        \Prob\left[\Gamma.\Decode(\sk,\tilde\cdw) \ne \tampered \wedge \tilde\cdw \ne \cdw\right] = \Prob\left[\Pi.\Detect(\kappa,\overline{x},\tilde\cdw) = \false \wedge \tilde\cdw \ne \cdw\right]  \ge 1/\poly(\secpar),
        \]
    \fi
    where \blue{$\tilde\cdw \getsr \hat f(\cdw)$} and where the probability is over the choice of the secret key, the randomness used to generate $\cdw \getsr \LLM_\kappa(\overline{x})$, and \blue{the randomness of $\hat f$}. This violates robustness of the watermarking scheme.

\end{proof}

When the output of the generative model is a \blue{fixed-length binary string (i.e., $\alphabet = \bits$ and $\cY = \bits^n$)}, the connection above becomes especially powerful. By mapping any robust watermarking scheme to a zero-bit tamper-detection code, we can directly apply our main impossibility result for codes to watermarking itself.

Specifically, our earlier theorem (\Cref{thm:impossible}) shows that no zero-bit tamper-detection code can be sound and robust to tampering beyond a certain threshold. By combining this with the reduction above, we obtain an immediate impossibility for robust watermarking of binary outputs:

\begin{coro}\label{coroll:impossibility-watermarking}
Let $\cModel = \{\LLM: \cX \rightarrow \blue{\bits^n}\}$ be any class of generative models with \blue{fixed-length binary outputs}.  
There is no watermarking scheme for $\cModel$ that achieves both soundness and $\hat\cF^{\mathrm{ind}}_\alpha$-robustness for any tampering rate $\alpha \geq (1+\delta)/2$, for any $\delta \in (0,1)$.  
\blue{This holds even if the model's prompt space $\cX$ is arbitrary.}\footnote{\blue{For completeness, we highlight that it is possible to extend~\Cref{coroll:impossibility-watermarking} to watermarking schemes supporting models with arbitrary output length (i.e., $\cModel = \{\LLM: \cX \rightarrow \blue{\bits^*}\}$). This is because~\Cref{thm:impossible} holds even for tamper-detection codes with variable codeword length. We refer the reader to ``Extension to weak soundness and variable length'' paragraph in~\Cref{sec:limit-no-msg}}.}
\end{coro}

\blue{The same discussion also suggests an analogous limitation under a constant-bounded watermark soundness notion for any fixed $c \in (0,1)$, in the same informal sense discussed after \Cref{thm:impossible}.} In particular, this rules out even a ``useless'' watermarking scheme whose soundness is achieved only with a high constant probability, such as $c=1/2$.

The impossibility threshold we have identified for tamper-resilient codes is not just a technical detail; it serves as a universal limit for watermarking \blue{with fixed-length binary outputs}. \blue{Any such watermarking scheme} must face this same upper bound, no matter how the scheme is constructed or how prompts are chosen. The recent work of Christ and Gunn~\cite{C:ChrGun24} provides a compelling example: their PRC-based watermarking schemes achieve robustness against independent bit-flip tampering up to the threshold $\alpha < (1-\delta)/2$, and their approach even accommodates certain types of deletions. Our corollary shows that this rate is not only achievable but also optimal. Although we do not claim a perfect correspondence between all watermarking schemes and all PRC constructions, since this depends on the precise model and tampering family, the upper bound on robustness holds for all such schemes.

\newcommand{\figFirstAttacksContent}{
    \centering
    \begin{minipage}{0.9\textwidth}
        \centering
        
        \begin{subfigure}[b]{0.3\linewidth}
            \centering
            \includegraphics[width=\textwidth]{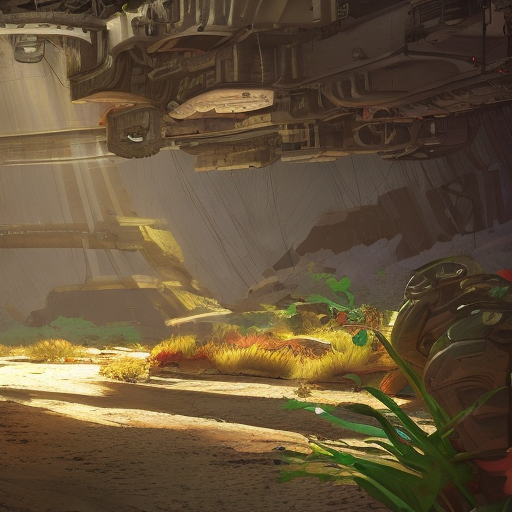}
            \caption{Original Image}
            \label{fig:five_images_a}
        \end{subfigure}
        \hfill
        \begin{subfigure}[b]{0.3\linewidth}
            \centering
            \includegraphics[width=\textwidth]{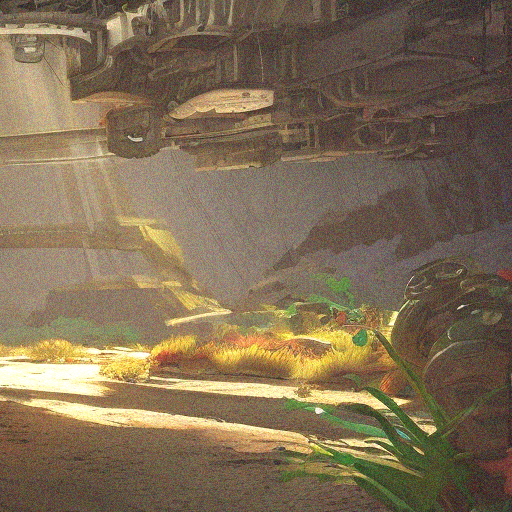}
            \caption{Color Attack: 23\% Error}
            \label{fig:five_images_b}
        \end{subfigure}
        \hfill
        \begin{subfigure}[b]{0.3\linewidth}
            \centering
            \includegraphics[width=\textwidth]{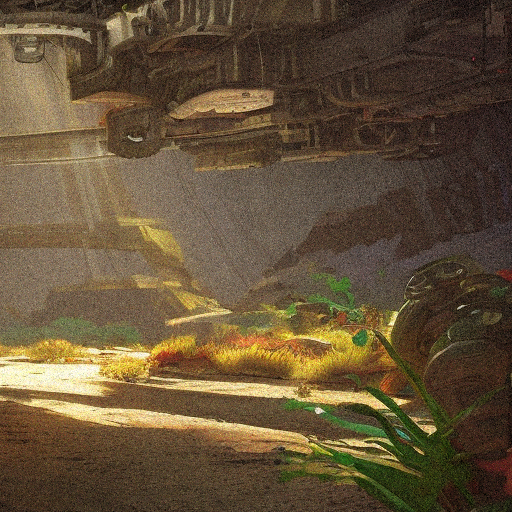}
            \caption{HSV Attack: 26\% Error}
            \label{fig:five_images_c}
        \end{subfigure}
        \hfill

        \vspace{0.5em}

        \hfill
        \begin{subfigure}[b]{0.48\linewidth}
            \centering
            \includegraphics[width=\textwidth]{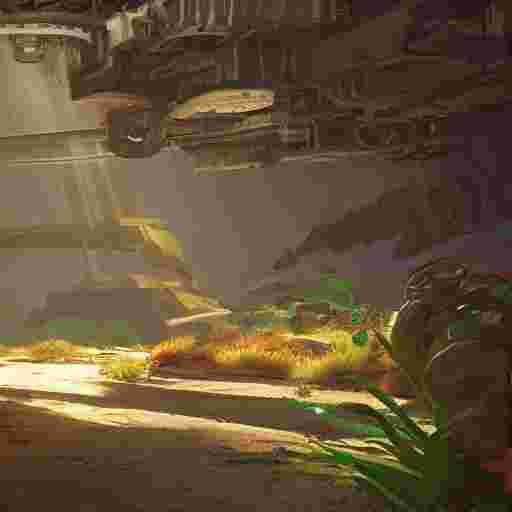}
            \caption{JPEG Attack: 32\% Error}
            \label{fig:five_images_d}
        \end{subfigure}
        \hfill
        \begin{subfigure}[b]{0.48\linewidth}
            \centering
            \includegraphics[width=\textwidth]{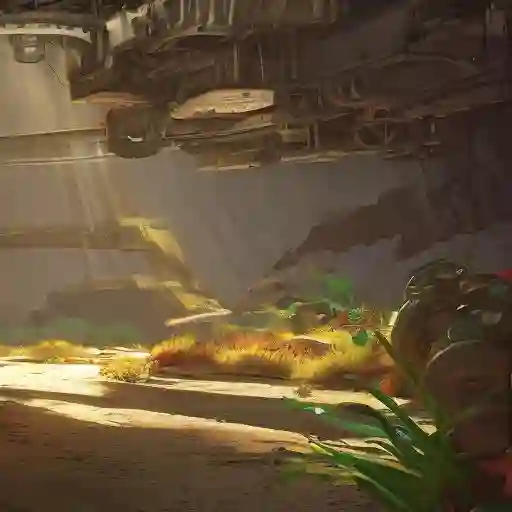}
            \caption{WebP Attack: 34\% Error}
            \label{fig:five_images_e}
        \end{subfigure}
        \hfill

        \caption{A series of images showing various attacks on a watermarked image, none of which removes the watermark}
        \label{fig:five_images_main}
    \end{minipage}
}

\newcommand{\figCropsContent}{
    \centering
    \begin{minipage}{0.8\textwidth} 
        \centering

        \begin{subfigure}[t]{0.24\linewidth}
            \centering
            \includegraphics[width=\linewidth]{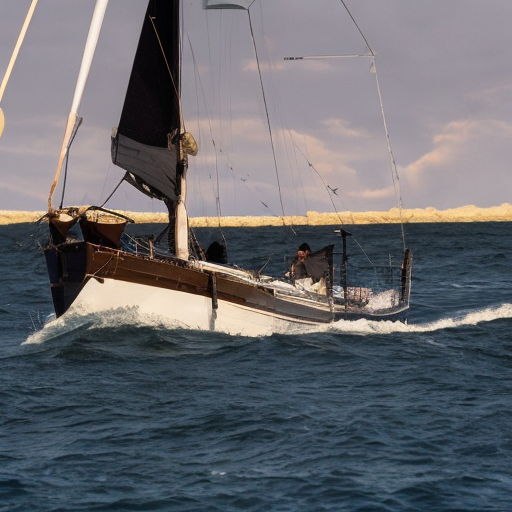}
            \caption{Original}
        \end{subfigure}
        \hfill
        \begin{subfigure}[t]{0.24\linewidth}
            \centering
            \includegraphics[width=\linewidth]{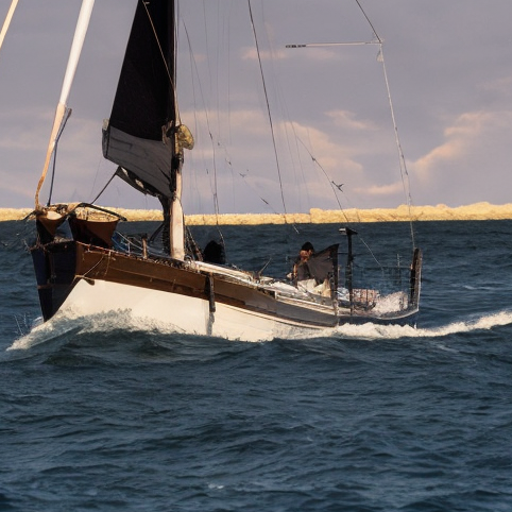}
            \caption{Crop \& Resize}
        \end{subfigure}
        \hfill
        \begin{subfigure}[t]{0.24\linewidth}
            \centering
            \includegraphics[width=\linewidth]{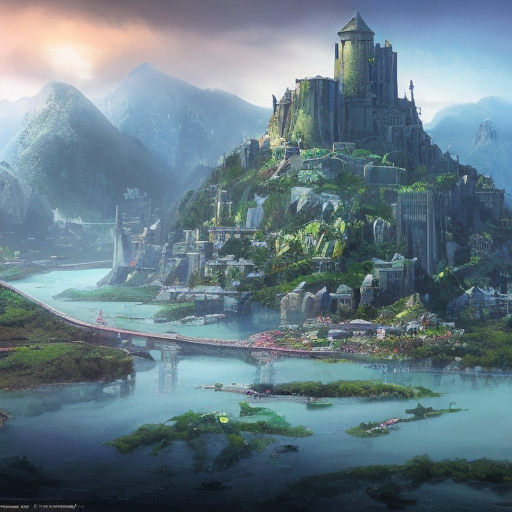}
            \caption{Original}
        \end{subfigure}
        \hfill
        \begin{subfigure}[t]{0.24\linewidth}
            \centering
            \includegraphics[width=\linewidth]{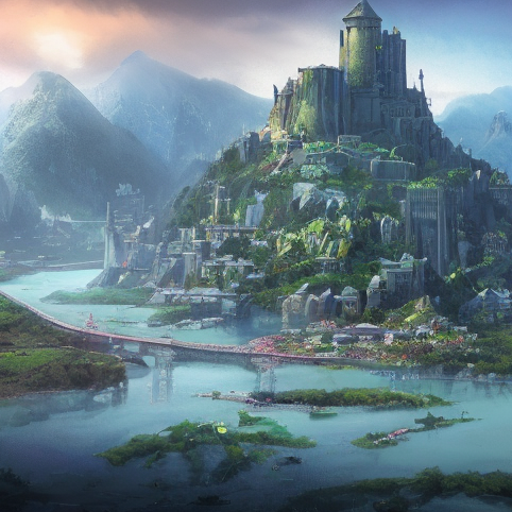}
            \caption{Crop \& Resize}
        \end{subfigure}

        \vspace{0.5em}

        \begin{subfigure}[t]{0.24\linewidth}
            \centering
            \includegraphics[width=\linewidth]{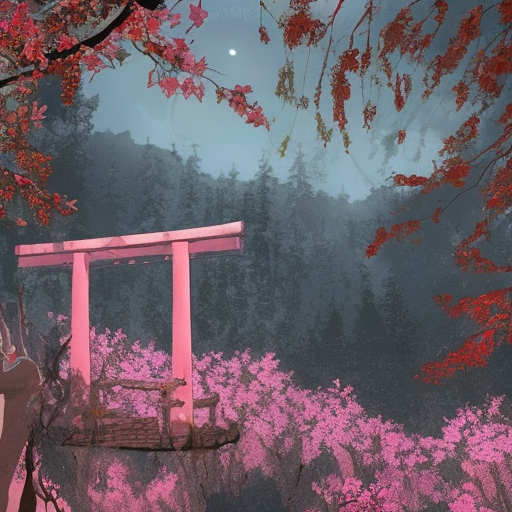}
            \caption{Original}
        \end{subfigure}
        \hfill
        \begin{subfigure}[t]{0.24\linewidth}
            \centering
            \includegraphics[width=\linewidth]{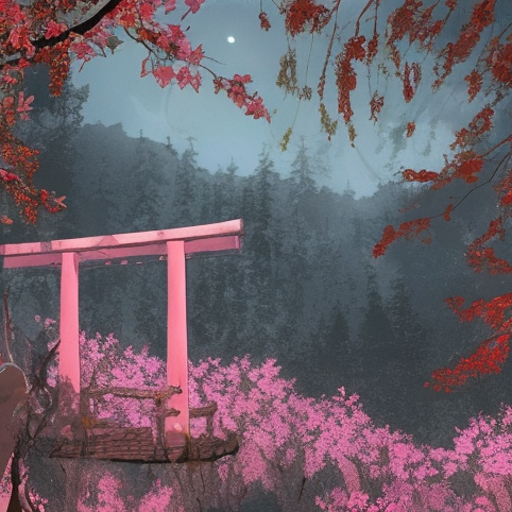}
            \caption{Crop \& Resize}
        \end{subfigure}
        \hfill
        \begin{subfigure}[t]{0.24\linewidth}
            \centering
            \includegraphics[width=\linewidth]{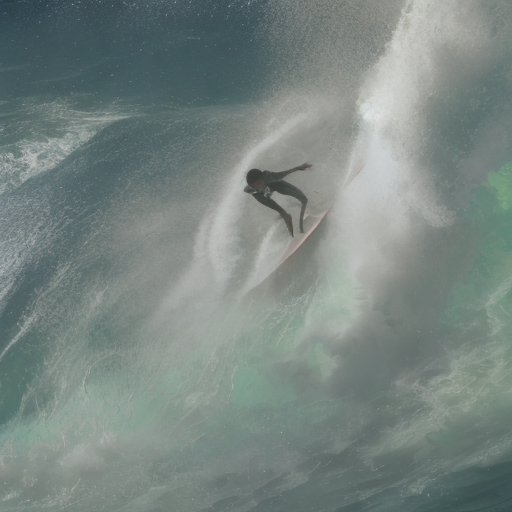}
            \caption{Original}
        \end{subfigure}
        \hfill
        \begin{subfigure}[t]{0.24\linewidth}
            \centering
            \includegraphics[width=\linewidth]{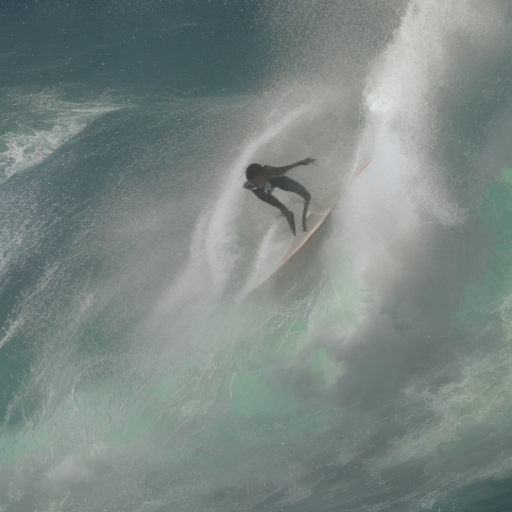}
            \caption{Crop \& Resize}
        \end{subfigure}

        \caption{Comparison between original and crop \& resize attacked images.}
        \label{fig:crop-resize-examples}
    \end{minipage}
}

\newcommand{\figLatentsContent}{
    \centering
    \includegraphics[width=0.9\textwidth]{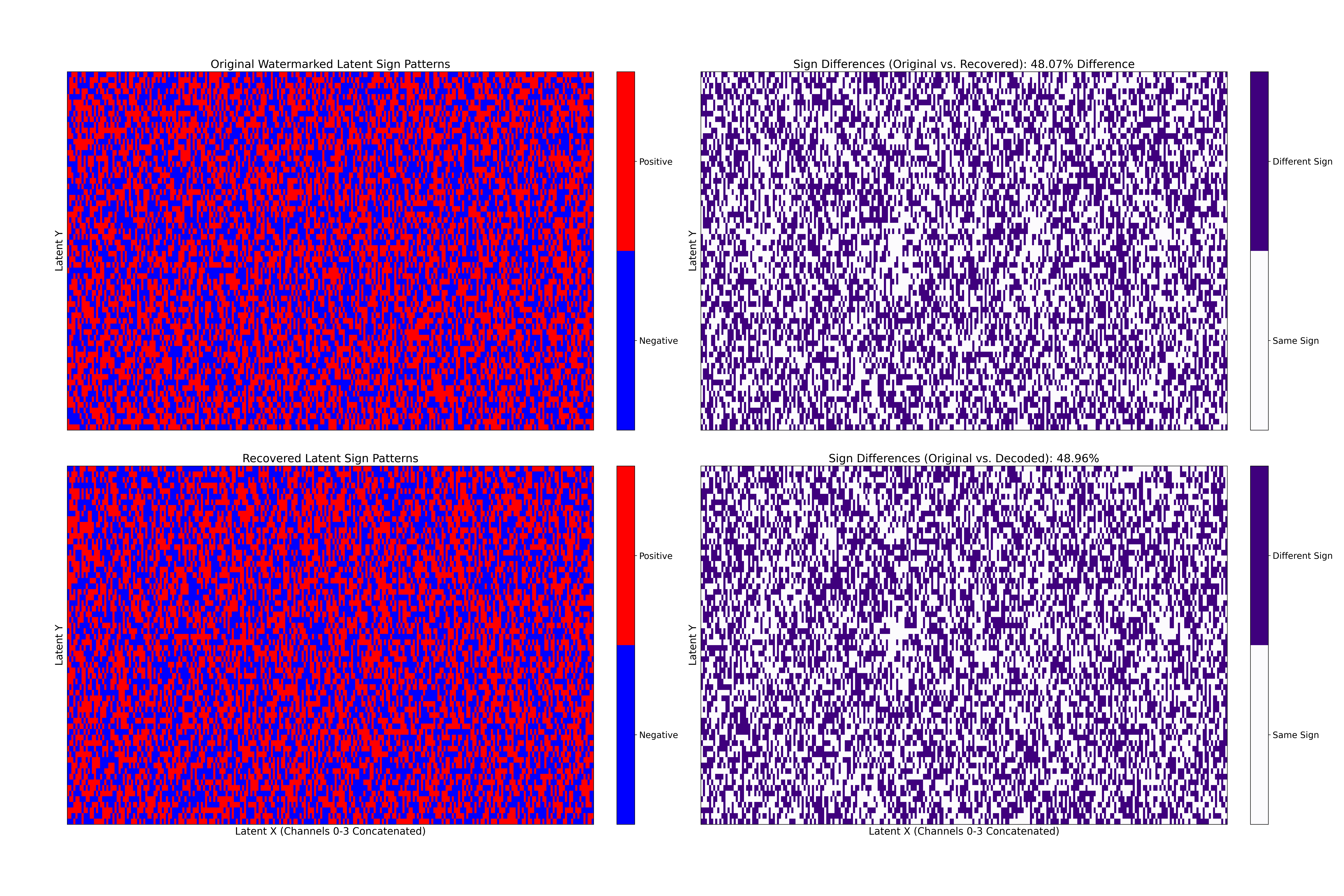}
    \caption{Latent sign analysis of a watermarked image before and after the crop and resize attack.
    Top left: the original latent sign pattern (red = positive, blue = negative) that encodes the watermark as a pseudorandom codeword.
    Bottom left: the latent recovered from the attacked image, which is visibly scrambled compared to the original.
    Top right: a difference map comparing the original and recovered latents, showing that $48.07\%$ of signs have flipped (purple = flipped, white = unchanged).
    Bottom right: a difference map after belief propagation, the error-correcting decoder, has tried to repair the errors.
    Because the error rate is already near the $50\%$ threshold, the decoder fails and converges to a different pseudorandom codeword.
    The resulting error rate increases slightly to $48.96\%$, confirming that the watermark cannot be recovered.}
    \label{fig:latent_sign_analysis_main}
}

\newcommand{\figVariantsContent}{
    \centering
    \begin{minipage}{0.9\textwidth}
        \centering

        \begin{subfigure}[t]{0.31\linewidth}
            \centering
            \includegraphics[width=\textwidth]{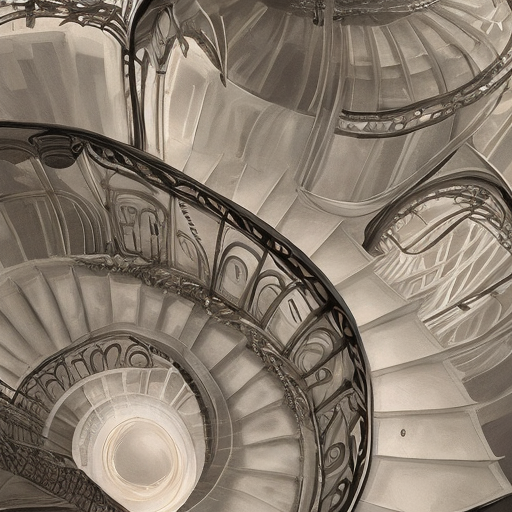}
            \caption{Original Image}
            \label{fig:orig}
        \end{subfigure}
        \hfill
        \begin{subfigure}[t]{0.31\linewidth}
            \centering
            \includegraphics[width=\textwidth]{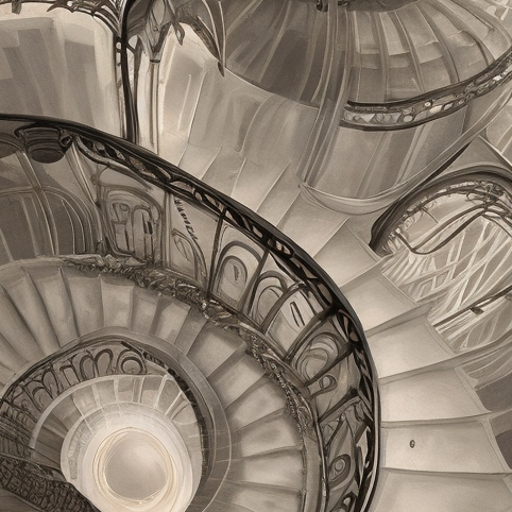}
            \caption{Crop \& Resize (47.2\% error)}
            \label{fig:crop-resize}
        \end{subfigure}
        \hfill
        \begin{subfigure}[t]{0.31\linewidth}
            \centering
            \includegraphics[width=\textwidth]{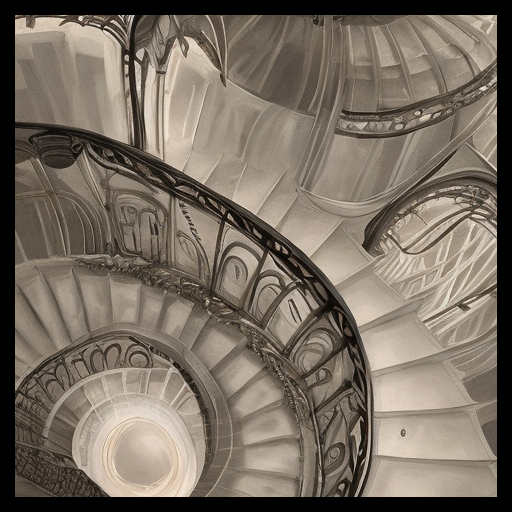}
            \caption{Crop \& Pad (16.7\% error)}
            \label{fig:crop-pad}
        \end{subfigure}
        
        \vspace{1em}

        \begin{subfigure}[t]{0.31\linewidth}
            \centering
            \includegraphics[width=\textwidth]{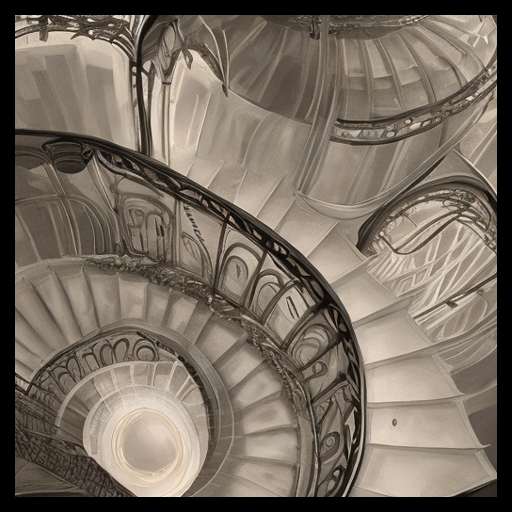}
            \caption{Downscale \& Pad (49.8\% error)}
            \label{fig:down-pad}
        \end{subfigure}
        \hspace{0.06\linewidth}
        \begin{subfigure}[t]{0.31\linewidth}
            \centering
            \includegraphics[width=\textwidth]{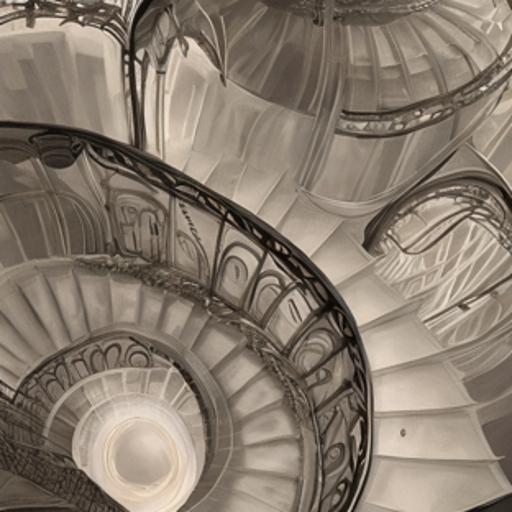}
            \caption{Downscale \& Upscale (12.1\% error)}
            \label{fig:down-up}
        \end{subfigure}

        \caption{Visualization of this branch of attacks (b-e) compared to the original (a)}
        \label{fig:multi_attack_visualization}
    \end{minipage}
}

\newcommand{\figLatentVariantsContent}{
    \centering
    \begin{minipage}{0.85\textwidth}
        \centering
        
        \begin{subfigure}[b]{0.48\linewidth}
            \centering
            \includegraphics[width=\linewidth]{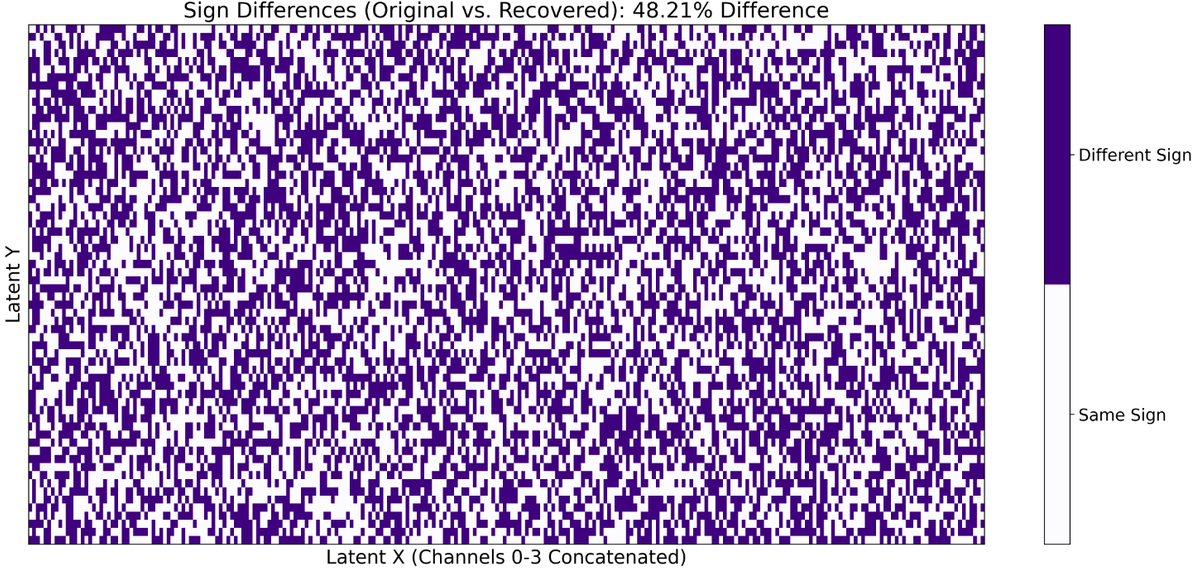}
            \caption{Crop \& Resize}
            \label{fig:sub_latent_a}
        \end{subfigure}
        \hfill
        \begin{subfigure}[b]{0.5\linewidth}
            \centering
            \includegraphics[width=\linewidth]{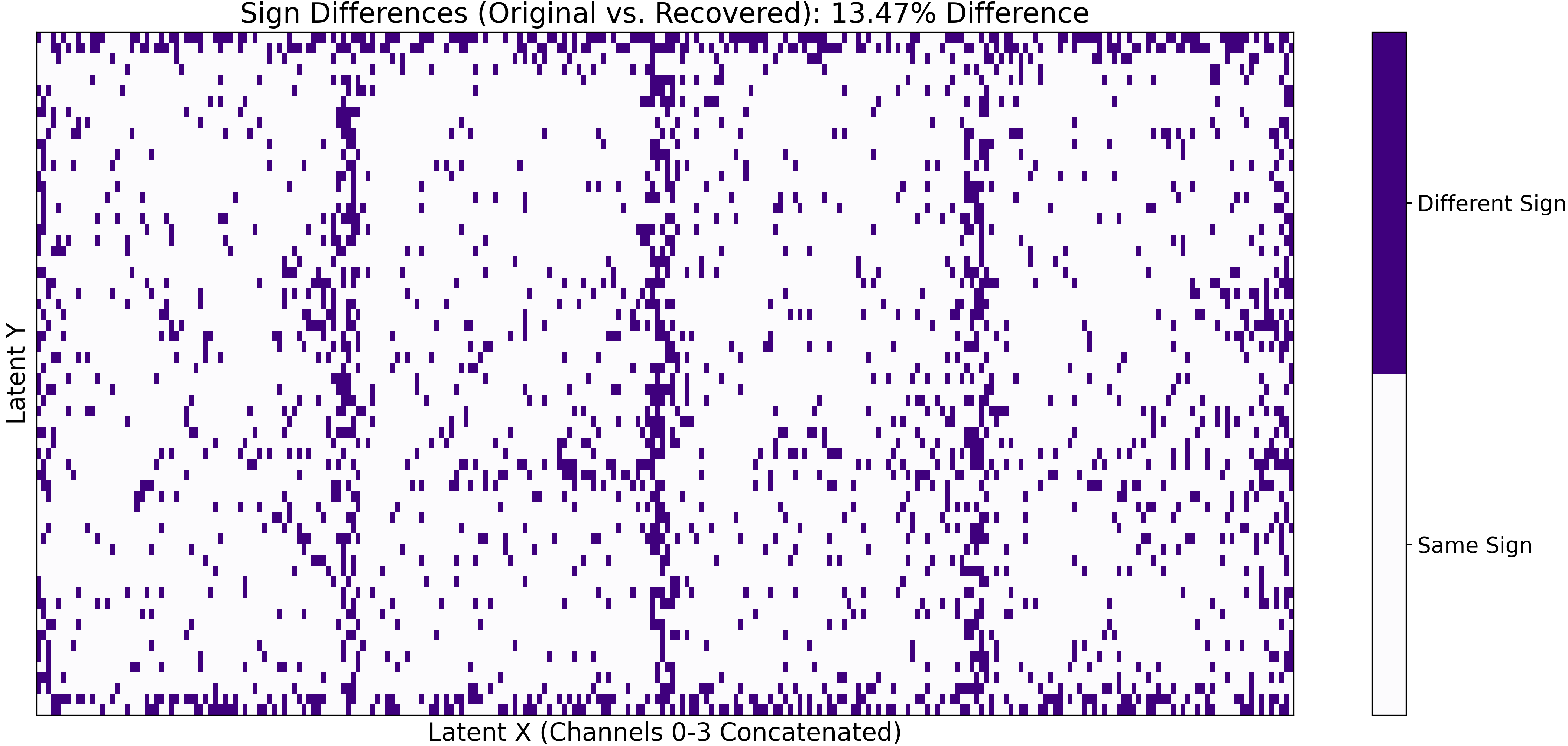}
            \caption{Crop \& Pad}
            \label{fig:sub_latent_b}
        \end{subfigure}

        \vspace{0.5em} 

        \begin{subfigure}[b]{0.5\linewidth}
            \centering
            \includegraphics[width=\linewidth]{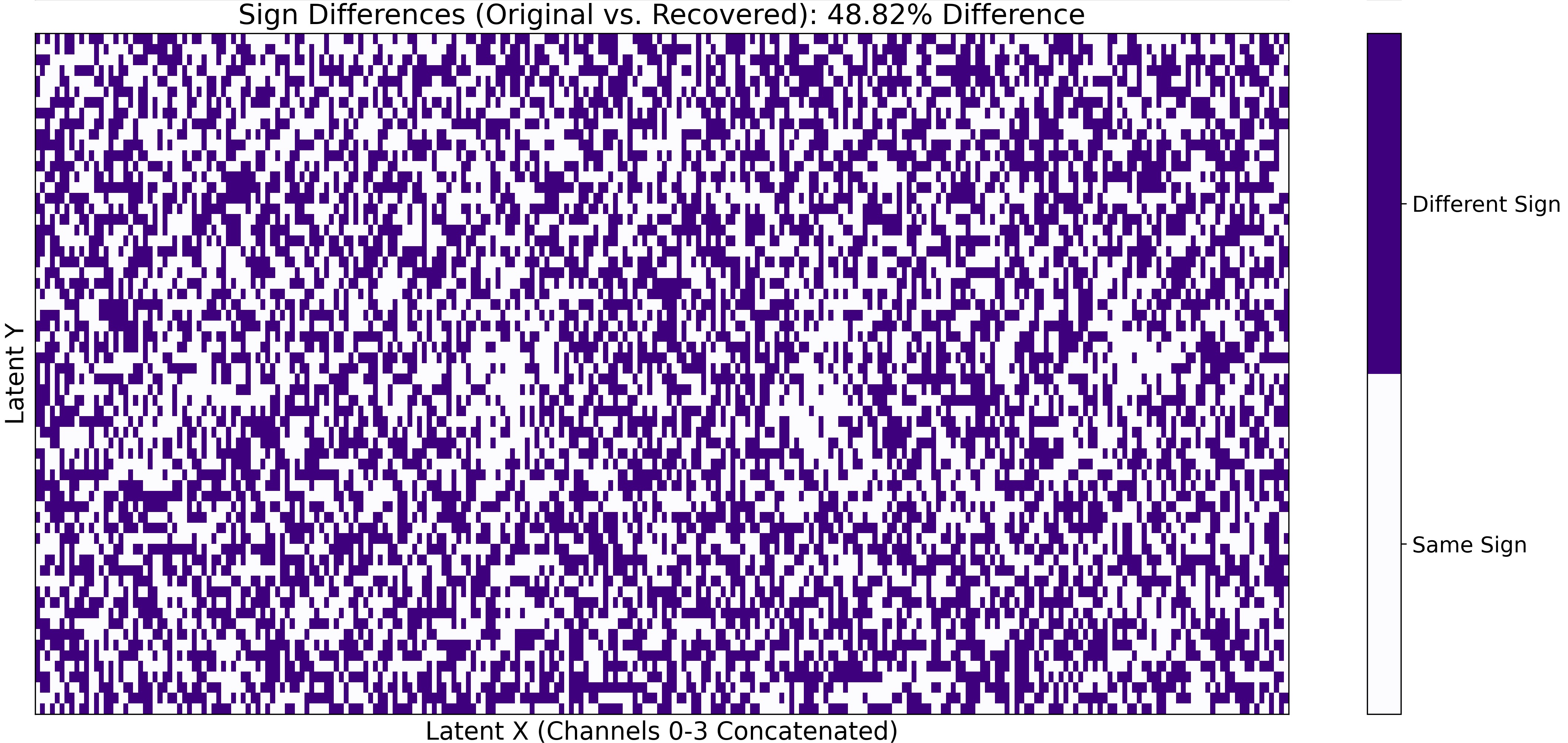}
            \caption{Downscale \& Pad}
            \label{fig:sub_latent_c}
        \end{subfigure}
        \hfill
        \begin{subfigure}[b]{0.48\linewidth}
            \centering
            \includegraphics[width=\linewidth]{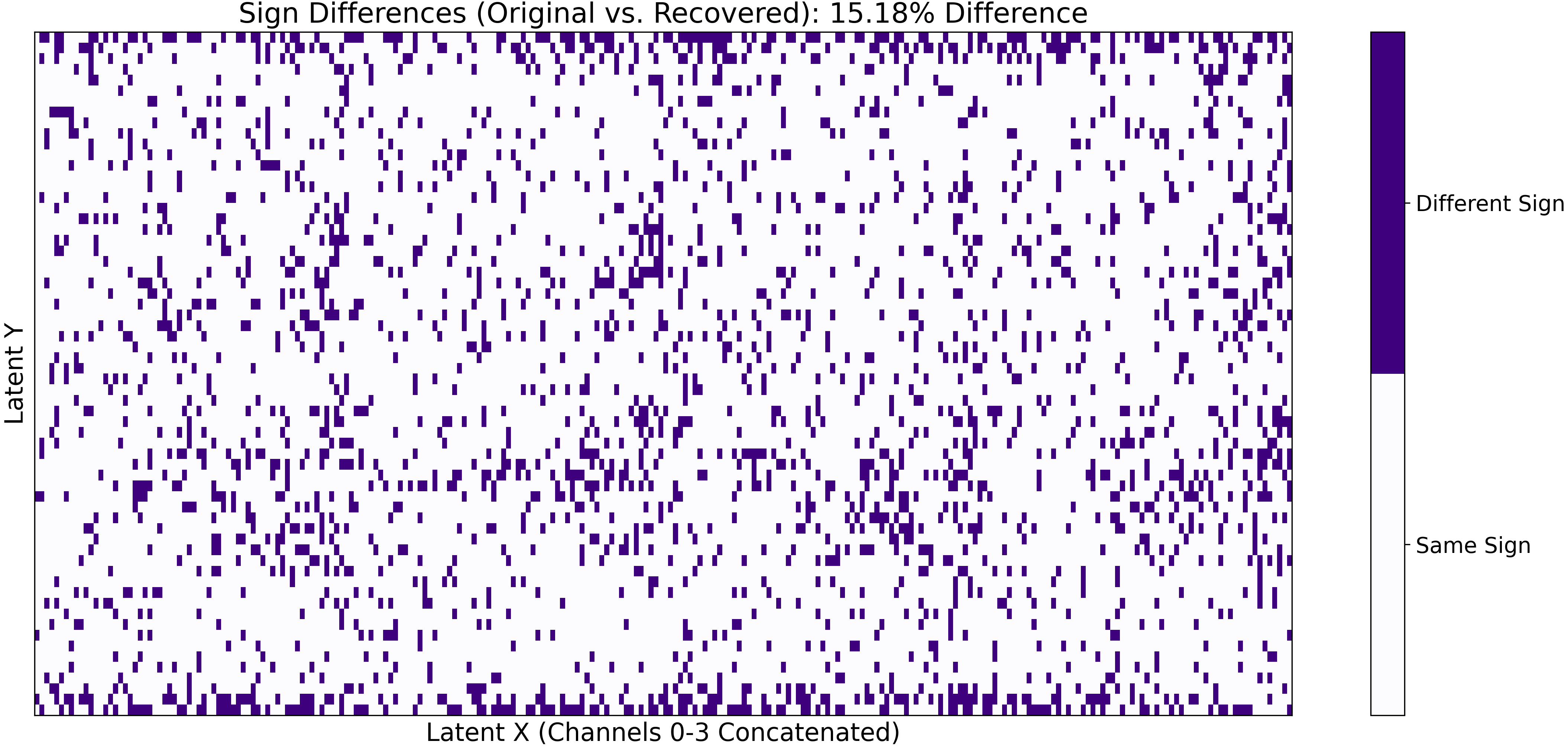}
            \caption{Downscale \& Upscale}
            \label{fig:sub_latent_d}
        \end{subfigure}

        \caption{Side-by-side comparison of latent space visualizations for the main attack and each variant}
        \label{fig:multi_attack_latent_analysis}
    \end{minipage}
}

\newcommand{\figBorderContent}{
    \centering
    \begin{minipage}{0.8\textwidth}
        \centering
        \begin{subfigure}[b]{0.48\linewidth}
            \centering
            \includegraphics[width=\linewidth]{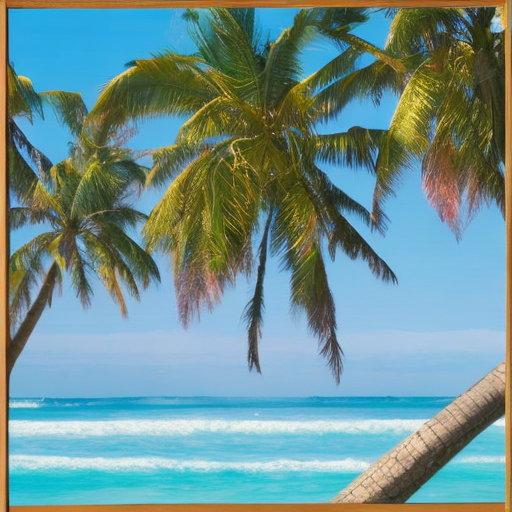}
            \caption{Original}
            \label{fig:original_border}
        \end{subfigure}
        \hfill
        \begin{subfigure}[b]{0.48\linewidth}
            \centering
            \includegraphics[width=\linewidth]{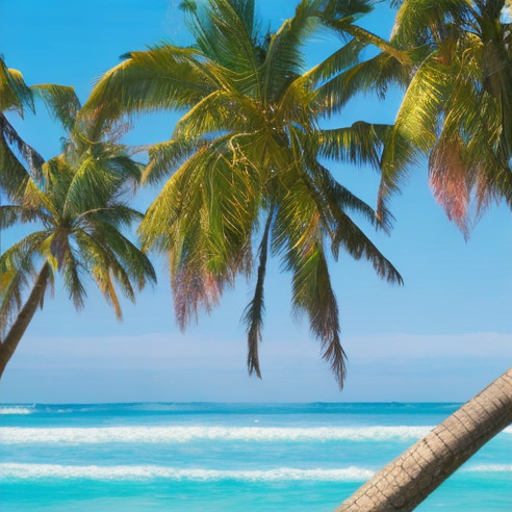}
            \caption{Cropped (10px)}
            \label{fig:attacked_border}
        \end{subfigure}
        
        \caption{Watermark removal on bordered image.}
        \label{fig:border_strategy}
    \end{minipage}
}

\section{A Concrete Attack}\label{sec:experiments}
\ifconf
\else
    \begin{figure}[t!]
        \figFirstAttacksContent
    \end{figure}
\fi

\ifconf
\else
\subsection{Setup}
\fi
Our goal is to test whether the theoretical robustness threshold we proved manifests in practice. 
For this purpose, we analyze the PRC watermark of Gunn, Zhao, and Song~\cite{gunn2025undetectable}, implemented in Stable Diffusion~2.1 Base (512$\times$512 resolution, 50 denoising steps). 
This watermark is representative of the state of the art: it preserves perceptual quality, is provably undetectable under standard assumptions, and resists a wide class of natural manipulations.

The scheme embeds a PRC (pseudorandom) codeword in the latent space. 
Let $\mathbf{v}\in \mathbb{R}^{4\times 64\times 64}$ denote the latent tensor produced by the diffusion model. 
\blue{For each entry $v_i$ of $\mathbf{v}$, the magnitude $|v_i|$ is preserved and the sign is set according to the corresponding codeword bit $\cdw_i\in\{0,1\}$, where $\cdw = \cdw_1||\ldots||\cdw_n$ is the bit decomposition of a PRC codeword.}
As a result, the watermarked latent still follows the Gaussian distribution required for high-quality image generation. 
The codeword is pseudorandom, so without the detection key the modified latent is indistinguishable from a fresh Gaussian sample.

\blue{Detection proceeds by first inverting a generated image $\hat{x}$ back into a latent $\hat{\mathbf{v}}$ using an approximate inversion procedure, and then extracting sign information from $\hat{\mathbf{v}}$.} 
Because inversion is lossy, $\hat{\mathbf{v}}$ typically agrees with the original latent on about $90\%$ of entries even in the absence of any edits. 
\blue{The scheme includes a detector and a separate belief-propagation (BP) decoding routine that attempts to recover the original codeword.} 
\blue{The key parameter in our experiments is the pre-decoding sign error rate, namely the fraction of latent coordinates whose sign differs between the original watermarked latent $\mathbf{v}$ and the latent $\hat{\mathbf{v}}$ recovered from the attacked image before decoding. The latent dimension is $n = 4 \cdot 64 \cdot 64 = 16384$.}
As reported in~\cite{gunn2025undetectable}, the BP decoder is probabilistic and exhibits variable robustness: it sometimes recovers the codeword under substantial perturbations, but overall it is less reliable than the detector, which remains robust across a wide range of attacks. 
In our experiments we observe that once the pre-decoding error rate approaches one-half of the latent positions, the BP decoder fails consistently.  

This threshold serves as the reference point for our experiments. 
The central question is whether simple image manipulations can drive the recovered latent to the $50\%$ error mark, thereby erasing the watermark.

\subsection{First Attempts}
Before developing our attack, we tested the robustness of~\cite{gunn2025undetectable} under a broad set of common image manipulations chosen to preserve image quality (including gaussian noise, gaussian blur, pixel-wise color shifts, and lossy format conversions).

\ifconf
\else
\paragraph{Image manipulations.}
\fi
Gaussian noise, Gaussian blur, HSV perturbations, and attribute adjustments (saturation, hue, exposure, contrast) were applied to watermarked images. 
Even when these visibly degraded the images, the pre-decoding error rate never exceeded $26\%$, and BP reduced the error to below $1\%$. 
The watermark remained detectable in every case.

\ifconf
\else
\paragraph{Pixel-wise color shifts.}
\fi
Next, we modified each pixel by adding fixed or randomized RGB offsets. 
A uniform shift of $(10,0,0)$, making the image slightly redder, produced under $10\%$ error. 
A large uniform shift of $(75,75,75)$, which washed out the image, raised the error to only $12\%$. 
Randomized RGB perturbations across all pixels, with values drawn uniformly between 0 and 50, produced slightly higher disruption but never more than $23\%$. 
These results suggest that uniform color manipulations do not significantly disturb the sign pattern in latent space.

\ifconf
\else
\paragraph{Lossy format conversions.}
\fi
Finally, we tested whether re-encoding the image in lossy formats could erase the watermark. 
JPEG compression at $15\%$ quality produced heavy visible degradation and a $32\%$ pre-decoding error rate, which BP reduced to $10\%$. 
WebP conversion at similar quality gave $34\%$ error, reduced to $15\%$ after decoding. 
In both cases, the watermark remained detectable.

\subsection{Our Successful Attack}

The failure of local manipulations suggests what is missing. 
Noise, blur, color shifts, and compression alter pixel values but preserve the coordinate system of the image. 
The encoder still sees essentially the same structure, and so the latent sign pattern (i.e., the PRC codeword $\cdw$) drifts only modestly. 
To break the watermark, we sought an operation that would force the encoder to reinterpret the image globally, while leaving the picture perceptually unchanged.

\ifconf
\else
    \begin{figure}[t!]
        \figCropsContent
    \end{figure}
\fi

\ifconf
\else
    \begin{figure}[t!]
        \figLatentsContent
    \end{figure}
\fi

This reasoning led to the crop-and-resize attack. 
We take a watermarked $512\times512$ image, crop $15$ pixels from each side to obtain a $482\times482$ image, and then resize it back to $512\times512$ using bicubic interpolation. 
The crop removes only a narrow border, and the interpolation smoothly reconstructs missing pixels from their neighbors. 
Visually, the result is nearly indistinguishable from the original. 
\blue{We use a $15$\,px crop by default because it is more reliable across examples, although for tightly framed images a $10$\,px crop still suffices, as in Figure~\ref{fig:border_strategy}. A $5$\,px crop is no longer reliable. Bilinear interpolation shows the same qualitative behavior as bicubic interpolation.}
Examples are shown in Figure~\ref{fig:crop-resize-examples}, where side-by-side comparisons reveal no perceptible difference between the original and the cropped-and-resized images.

The effect on the watermark is decisive. 
Across more than one thousand test cases, the recovered latent after crop-and-resize showed pre-decoding error rates concentrated near $50\%$. 
This error rate is exactly the threshold at which BP decoding fails. 
Indeed, in every trial the decoder was unable to recover the codeword, and detection failed completely. 
Unlike the earlier attacks, which either left the watermark intact or visibly degraded the image, this single transformation erased the watermark while keeping the content unchanged.

In short, crop-and-resize is the only edit we found that preserves perceptual quality while driving the latent representation to the theoretical boundary. 
The watermark disappears, not because the image is damaged, but because the latent has been resampled in a way that overwhelms the error-correcting capacity of the code.

\ifconf
\else
    \begin{figure}[t!]
        \figVariantsContent
    \end{figure}
\fi

\ifconf
\else
\subsection{Why It Works}
\fi
The success of crop-and-resize can be traced to how the encoder interprets the image. 
Local edits (noise, blur, color shifts, compression) perturb pixel values but leave the grid of the image intact. 
The encoder still recognizes the same arrangement, and the latent sign pattern changes only in scattered places. 
Cropping followed by resizing is different. 
Cropping removes a thin border, and resizing does not simply stretch the remaining pixels. 
It recomputes every pixel value by bicubic interpolation on a new lattice. 
The encoder is then asked to describe the same visual content in a new coordinate system. 
This global resampling is what drives nearly half of the latent entries across zero, flipping their signs.

\ifconf
\else
\paragraph{Latent-space effects.}
\fi
Figure~\ref{fig:latent_sign_analysis_main} shows this effect in detail. 
The original latent sign pattern (top left) encodes the watermark (i.e., the PRC codeword) as a pseudorandom arrangement of signs. 
After crop-and-resize, the recovered latent (bottom left) is visibly scrambled. 
The difference map before decoding (top right) shows that $48.07\%$ of the signs have flipped. 
This is essentially the information-theoretic boundary: once errors approach $50\%$, the codeword is indistinguishable from random. 
\blue{Of course, crop-and-resize is a structured image transformation rather than the i.i.d. corruption channel of \Cref{thm:impossible}. The point is that, empirically, it pushes the recovered sign pattern close to the same critical threshold at which the codeword becomes effectively indistinguishable from random.}
The bottom-right panel shows the result after belief propagation, the error-correcting decoder, has attempted to repair the errors. 
At this level of corruption, BP cannot recover the original codeword and instead converges to a different pseudorandom codeword. 
As a result, the measured error rises slightly to $48.96\%$. 
This increase is not paradoxical but diagnostic: the decoder has lost all correlation with the true codeword. 
The image itself is unchanged, but the watermark is irretrievable.

\ifconf
\else
\paragraph{Variants.}
\fi

\ifconf
\else
    \begin{figure}[t!]
        \figLatentVariantsContent
    \end{figure}
\fi

Other global manipulations show the same pattern. Figure~\ref{fig:multi_attack_visualization} compares crop and resize with several related operations. 
Cropping and padding with black pixels produced about $16.7\%$ error and left a conspicuous border. 
Downscaling and then upscaling, even aggressively to $312\times312$ before resizing back, produced about $12.1\%$ error. 
Neither came close to the threshold. 
Downscaling and padding with black pixels did reach $\sim 50\%$ error (Figure~\ref{fig:multi_attack_latent_analysis}), but the added border made the alteration visually obvious, as shown in subfigure~(d) of Figure~\ref{fig:multi_attack_visualization}.
In contrast, crop and resize combined two properties: it perturbed the latent enough to destroy the watermark, and it did so while preserving the appearance of the image.

\ifconf
\else
\paragraph{Generalization.}
\fi
The vulnerability is not limited to this specific setup. 
Generative models often allow users to control framing explicitly. 
An adversary can request an image with an artificial border and then remove it by cropping and resizing. 
Figure~\ref{fig:border_strategy} illustrates this approach: the model produces an image with a frame, the border is removed, and the watermark disappears while the content remains intact. 
No knowledge of the watermark key is needed, and no optimization is required. 
The broader lesson is that transformations that globally resample the image while keeping their perceptual content largely unchanged can push the latent toward the robustness threshold.

\section*{Data and Code Availability}

Our theoretical results are self-contained. The empirical study in Section~\ref{sec:experiments} uses only synthetic images generated by Stable Diffusion together with the publicly described PRC-based watermark of Gunn, Zhao, and Song~\cite{gunn2025undetectable}. No proprietary data or data from human subjects is used.

To support reproducibility, we provide a repository at \url{https://github.com/ygoonati/Robust-Watermarking-Limits}.
The repository contains:

\begin{itemize}
  \item the code that interfaces with the watermark implementation and applies the crop-and-resize transformation;
  \item configuration files for the Stable Diffusion pipeline and prompts used in our experiments;
  \item scripts that reproduce the figures and numerical results in the paper from a fresh checkout.
\end{itemize}

These artifacts are sufficient to re-run the experiments and verify the empirical claims in the paper.


\ifconf
\else
    \begin{figure}[t]
        \figBorderContent
    \end{figure}
\fi

\section*{Acknowledgments}
Danilo Francati was supported by project SERICS (PE00000014), under the
MUR National Recovery and Resilience Plan funded by the European Union-
NextGenerationEU. 
Daniele Venturi is member of the Gruppo Nazionale Calcolo
Scientifico Istituto Nazionale di Alta Matematica (GNCS-INdAM). His research was partially supported by project SERICS (PE00000014) and by project PARTHENON (B53D23013000006)---under the MUR National Recovery and Resilience Plan funded by the European Union-NextGenerationEU---by project BEAT---funded by Sapienza University of Rome---and by the AI-Sec Lab initiative between Sapienza University of Rome and the CYBER 4.0 Competence Center.

\FloatBarrier
\ifconf
\else
	\bibliographystyle{alpha}
\fi

\bibliography{references,abbrev3.bib,crypto.bib}

@string{ieee =                  {IEEE}}

@string{springer =              "Springer"}

@string{ton =                   "{IEEE/ACM} Transactions on Networking"}

@string{acm =                   "{ACM}"}

@misc{Aaronson22,
  author       = {Scott Aaronson},
  title        = {My {AI} Safety Lecture for {UT} Effective Altruism},
  howpublished = {\url{https://scottaaronson.blog/?p=6823}},
  note         = {Discusses watermarking projects at OpenAI. Accessed: September 2025},
  year         = {2022}
}

@misc{US23,
  author       = {{Executive Office of the President of the United States}},
  title        = {{Executive Order No. 14110: Safe, Secure, and Trustworthy Development and Use of Artificial Intelligence}},
  howpublished = {\url{https://www.federalregister.gov/documents/2023/11/01/2023-24283}},
  note         = {Federal Register, Vol.\ 88, No.\ 212, pp.\ 75191--75212. Accessed: September 2025},
  year         = {2023}
}

@misc{EU24,
  author       = {{European Parliament and Council of the European Union}},
  title        = {{Regulation (EU) 2024/1689: Artificial Intelligence Act}},
  howpublished = {\url{https://eur-lex.europa.eu/eli/reg/2024/1689/oj/eng}},
  note         = {Official Journal of the European Union, L 1689, 12 July 2024. Accessed: September 2025},
  year         = {2024}
}

@inproceedings{gunn2025undetectable,
  title     = {An Undetectable Watermark for Generative Image Models},
  author    = {Sam Gunn and Xuandong Zhao and Dawn Song},
  booktitle = {The Thirteenth International Conference on Learning Representations},
  year      = {2025},
  url       = {https://openreview.net/forum?id=jlhBFm7T2J}
}

@inproceedings{unmarker2025,
  title   = {UnMarker: A Universal Attack on Defensive Image Watermarking},
  author  = {Kassis, Andre and Hengartner, Urs},
  booktitle = {2025 IEEE Symposium on Security and Privacy (SP)},
  year    = {2025},
  doi     = {10.1109/SP61157.2025.00005}
}

@article{zhang2023watermarks,
  title   = {Watermarks in the sand: Impossibility of strong watermarking for generative models},
  author  = {Zhang, Hanlin and Edelman, Benjamin L and Francati, Danilo and Venturi, Daniele and Ateniese, Giuseppe and Barak, Boaz},
  journal = {arXiv preprint arXiv:2311.04378},
  year    = {2023}
}

@inproceedings{zhang2024watermarks,
  title     = {Watermarks in the sand: impossibility of strong watermarking for language models},
  author    = {Zhang, Hanlin and Edelman, Benjamin L and Francati, Danilo and Venturi, Daniele and Ateniese, Giuseppe and Barak, Boaz},
  booktitle = {Forty-first International Conference on Machine Learning},
  year      = {2024}
}

@inproceedings{boland_95,
  author    = {Boland, F.M. and O'Ruanaidh, J.J.K. and Dautzenberg, C.},
  booktitle = {Fifth International Conference on Image Processing and its Applications, 1995.},
  title     = {Watermarking digital images for copyright protection},
  year      = {1995},
  volume    = {},
  number    = {},
  pages     = {326-330},
  doi       = {10.1049/cp:19950674}
}

@book{cox_et_al_07,
  author    = {Cox, Ingemar and Miller, Matthew and Bloom, Jeffrey and Fridrich, Jessica and Kalker, Ton},
  title     = {Digital Watermarking and Steganography},
  year      = {2007},
  isbn      = {9780080555805},
  publisher = {Morgan Kaufmann Publishers Inc.}
}

@inproceedings{Hayes_et_al_17,
  author    = {Hayes, Jamie and Danezis, George},
  title     = {Generating steganographic images via adversarial training},
  year      = {2017},
  isbn      = {9781510860964},
  publisher = {Curran Associates Inc.},
  address   = {Red Hook, NY, USA},
  booktitle = {Proceedings of the 31st International Conference on Neural Information Processing Systems},
  pages     = {1951–1960},
  numpages  = {10},
  location  = {Long Beach, California, USA},
  series    = {NIPS'17}
}

@article{Ruanaidh_1997,
  title   = {Rotation, scale and translation invariant digital image watermarking},
  author  = {Joseph {\'O} Ruanaidh and Thierry Pun},
  journal = {Proceedings of International Conference on Image Processing},
  year    = {1997},
  volume  = {1},
  pages   = {536-539 vol.1},
  url     = {https://api.semanticscholar.org/CorpusID:2678473}
}

@inproceedings{Zhu_18,
  author    = {Zhu, Jiren and Kaplan, Russell and Johnson, Justin and Fei-Fei, Li},
  title     = {HiDDeN: Hiding Data With Deep Networks},
  year      = {2018},
  isbn      = {978-3-030-01266-3},
  publisher = {Springer-Verlag},
  address   = {Berlin, Heidelberg},
  url       = {https://doi.org/10.1007/978-3-030-01267-0_40},
  doi       = {10.1007/978-3-030-01267-0_40},
  booktitle = {Computer Vision – ECCV 2018: 15th European Conference, Munich, Germany, September 8-14, 2018, Proceedings, Part XV},
  pages     = {682–697},
  numpages  = {16},
  location  = {Munich, Germany}
}

@inproceedings{arnold_00,
  title        = {Audio watermarking: Features, applications and algorithms},
  author       = {Arnold, Michael},
  booktitle    = {2000 IEEE International conference on multimedia and expo. ICME2000. Proceedings. Latest advances in the fast changing world of multimedia (cat. no. 00TH8532)},
  volume       = {2},
  pages        = {1013--1016},
  year         = {2000},
  organization = {IEEE}
}

@inproceedings{Boney_96,
  title        = {Digital watermarks for audio signals},
  author       = {Boney, Laurence and Tewfik, Ahmed H and Hamdy, Khaled N},
  booktitle    = {Proceedings of the third IEEE international conference on multimedia computing and systems},
  pages        = {473--480},
  year         = {1996},
  organization = {IEEE}
}

@article{Wan_22,
  title    = {A comprehensive survey on robust image watermarking},
  journal  = {Neurocomputing},
  volume   = {488},
  pages    = {226-247},
  year     = {2022},
  issn     = {0925-2312},
  doi      = {10.1016/j.neucom.2022.02.083},
  url      = {https://www.sciencedirect.com/science/article/pii/S0925231222002533},
  author   = {Wenbo Wan and Jun Wang and Yunming Zhang and Jing Li and Hui Yu and Jiande Sun}
}

@inproceedings{An_24,
  author    = {An, Bang and Ding, Mucong and Rabbani, Tahseen and Agrawal, Aakriti and Xu, Yuancheng and Deng, Chenghao and Zhu, Sicheng and Mohamed, Abdirisak and Wen, Yuxin and Goldstein, Tom and Huang, Furong},
  title     = {WAVES: benchmarking the robustness of image watermarks},
  year      = {2024},
  publisher = {JMLR.org},
  booktitle = {Proceedings of the 41st International Conference on Machine Learning},
  articleno = {59},
  numpages  = {37},
  location  = {Vienna, Austria},
  series    = {ICML'24}
}

@inproceedings{Yu_21,
  author    = {Yu, Ning and Skripniuk, Vladislav and Abdelnabi, Sahar and Fritz, Mario},
  booktitle = {2021 IEEE/CVF International Conference on Computer Vision (ICCV)},
  title     = {Artificial Fingerprinting for Generative Models: Rooting Deepfake Attribution in Training Data},
  year      = {2021},
  volume    = {},
  number    = {},
  pages     = {14428-14437},
  doi       = {10.1109/ICCV48922.2021.01418}
}

@inproceedings{Kirchenbauer_23,
  title     = {A Watermark for Large Language Models},
  author    = {Kirchenbauer, John and Geiping, Jonas and Wen, Yuxin and Katz, Jonathan and Miers, Ian and Goldstein, Tom},
  booktitle = {Proceedings of the 40th International Conference on Machine Learning},
  pages     = {17061--17084},
  year      = {2023},
  editor    = {Krause, Andreas and Brunskill, Emma and Cho, Kyunghyun and Engelhardt, Barbara and Sabato, Sivan and Scarlett, Jonathan},
  volume    = {202},
  series    = {Proceedings of Machine Learning Research},
  month     = {23--29 Jul},
  publisher = {PMLR},
  url       = {https://proceedings.mlr.press/v202/kirchenbauer23a.html}
}

@inproceedings{Fernandez_23,
  title     = {The stable signature: Rooting watermarks in latent diffusion models},
  author    = {Fernandez, Pierre and Couairon, Guillaume and J{\'e}gou, Herv{\'e} and Douze, Matthijs and Furon, Teddy},
  booktitle = {Proceedings of the IEEE/CVF International Conference on Computer Vision},
  pages     = {22466--22477},
  year      = {2023}
}

@article{kuditipudi_24,
  title   = {Robust Distortion-free Watermarks for Language Models},
  author  = {Rohith Kuditipudi and John Thickstun and Tatsunori Hashimoto and Percy Liang},
  journal = {Transactions on Machine Learning Research},
  issn    = {2835-8856},
  year    = {2024},
  url     = {https://openreview.net/forum?id=FpaCL1MO2C},
  note    = {}
}

@inproceedings{Christ_23,
  title     = {Undetectable Watermarks for Language Models},
  author    = {Christ, Miranda and Gunn, Sam and Zamir, Or},
  booktitle = {Proceedings of Thirty Seventh Conference on Learning Theory},
  pages     = {1125--1139},
  year      = {2024},
  editor    = {Agrawal, Shipra and Roth, Aaron},
  volume    = {247},
  series    = {Proceedings of Machine Learning Research},
  month     = {30 Jun--03 Jul},
  publisher = {PMLR},
  url       = {https://proceedings.mlr.press/v247/christ24a.html}
}

@article{ghentiyala_24,
  title   = {New constructions of pseudorandom codes},
  author  = {Ghentiyala, Surendra and Guruswami, Venkatesan},
  journal = {arXiv preprint arXiv:2409.07580},
  year    = {2024}
}

@article{garg23,
  title   = {Publicly-Detectable Watermarking for Language Models},
  author  = {Fairoze, Jaiden and Garg, Sanjam and Jha, Somesh and Mahloujifar, Saeed and Mahmoody, Mohammad and Wang, Mingyuan},
  journal = {IACR Communications in Cryptology},
  volume  = {1},
  number  = {4},
  year    = {2025}
}

@inproceedings{paraphrasing_attack_1,
  title     = {Revealing Weaknesses in Text Watermarking Through Self-Information Rewrite Attacks},
  author    = {Cheng, Yixin and Guo, Hongcheng and Li, Yangming and Sigal, Leonid},
  booktitle = {Forty-second International Conference on Machine Learning},
  year      = {2025}
}

@article{paraphrasing_attack_2,
  title   = {Paraphrasing evades detectors of ai-generated text, but retrieval is an effective defense},
  author  = {Krishna, Kalpesh and Song, Yixiao and Karpinska, Marzena and Wieting, John and Iyyer, Mohit},
  journal = {Advances in Neural Information Processing Systems},
  volume  = {36},
  pages   = {27469--27500},
  year    = {2023}
}

@inproceedings{lukas_24,
  title     = {Leveraging Optimization for Adaptive Attacks on Image Watermarks},
  author    = {Nils Lukas and Abdulrahman Diaa and Lucas Fenaux and Florian Kerschbaum},
  booktitle = {The Twelfth International Conference on Learning Representations},
  year      = {2024},
  url       = {https://openreview.net/forum?id=O9PArxKLe1}
}

@inproceedings{saberi_24,
  title     = {Robustness of {AI}-Image Detectors: Fundamental Limits and Practical Attacks},
  author    = {Mehrdad Saberi and Vinu Sankar Sadasivan and Keivan Rezaei and Aounon Kumar and Atoosa Chegini and Wenxiao Wang and Soheil Feizi},
  booktitle = {The Twelfth International Conference on Learning Representations},
  year      = {2024},
  url       = {https://openreview.net/forum?id=dLoAdIKENc}
}

@inproceedings{Zhao_24b,
  author    = {Zhao, Xuandong and Zhang, Kexun and Su, Zihao and Vasan, Saastha and Grishchenko, Ilya and Kruegel, Christopher and Vigna, Giovanni and Wang, Yu-Xiang and Li, Lei},
  booktitle = {Advances in Neural Information Processing Systems},
  editor    = {A. Globerson and L. Mackey and D. Belgrave and A. Fan and U. Paquet and J. Tomczak and C. Zhang},
  pages     = {8643--8672},
  publisher = {Curran Associates, Inc.},
  title     = {Invisible Image Watermarks Are Provably Removable Using Generative AI},
  url       = {https://proceedings.neurips.cc/paper_files/paper/2024/file/10272bfd0371ef960ec557ed6c866058-Paper-Conference.pdf},
  volume    = {37},
  year      = {2024}
}

@inproceedings{Fei_22,
  author    = {Fei, Jianwei and Xia, Zhihua and Tondi, Benedetta and Barni, Mauro},
  booktitle = {2022 IEEE International Workshop on Information Forensics and Security (WIFS)},
  title     = {Supervised GAN Watermarking for Intellectual Property Protection},
  year      = {2022},
  volume    = {},
  number    = {},
  pages     = {1-6},
  doi       = {10.1109/WIFS55849.2022.9975409}
}

@article{Zeng_23,
  title   = {Securing deep generative models with universal adversarial signature},
  author  = {Zeng, Yu and Zhou, Mo and Xue, Yuan and Patel, Vishal M},
  journal = {arXiv preprint arXiv:2305.16310},
  year    = {2023}
}

@article{Zhao_23c,
  title   = {A recipe for watermarking diffusion models},
  author  = {Zhao, Yunqing and Pang, Tianyu and Du, Chao and Yang, Xiao and Cheung, Ngai-Man and Lin, Min},
  journal = {arXiv preprint arXiv:2303.10137},
  year    = {2023}
}

@inproceedings{STOC:AACDG25,
  author    = {Alrabiah, Omar and Ananth, Prabhanjan and Christ, Miranda and Dodis, Yevgeniy and Gunn, Sam},
  title     = {Ideal Pseudorandom Codes},
  year      = {2025},
  isbn      = {9798400715105},
  publisher = {Association for Computing Machinery},
  address   = {New York, NY, USA},
  url       = {https://doi.org/10.1145/3717823.3718309},
  doi       = {10.1145/3717823.3718309},
  booktitle = {Proceedings of the 57th Annual ACM Symposium on Theory of Computing},
  pages     = {1638–1647},
  numpages  = {10},
  location  = {Prague, Czechia},
  series    = {STOC '25}
}

@article{wen_2023,
  title   = {Tree-ring watermarks: Fingerprints for diffusion images that are invisible and robust},
  author  = {Wen, Yuxin and Kirchenbauer, John and Geiping, Jonas and Goldstein, Tom},
  journal = {arXiv preprint arXiv:2305.20030},
  year    = {2023}
}

@inproceedings{text_1,
  author    = {Atallah, Mikhail J.
               and Raskin, Victor
               and Crogan, Michael
               and Hempelmann, Christian
               and Kerschbaum, Florian
               and Mohamed, Dina
               and Naik, Sanket},
  editor    = {Moskowitz, Ira S.},
  title     = {Natural Language Watermarking: Design, Analysis, and a Proof-of-Concept Implementation},
  booktitle = {Information Hiding},
  year      = {2001},
  publisher = {Springer Berlin Heidelberg},
  address   = {Berlin, Heidelberg},
  pages     = {185--200},
  isbn      = {978-3-540-45496-0}
}

@inproceedings{text_2,
  author    = {Atallah, Mikhail J.
               and Raskin, Victor
               and Hempelmann, Christian F.
               and Karahan, Mercan
               and Sion, Radu
               and Topkara, Umut
               and Triezenberg, Katrina E.},
  editor    = {Petitcolas, Fabien A. P.},
  title     = {Natural Language Watermarking and Tamperproofing},
  booktitle = {Information Hiding},
  year      = {2003},
  publisher = {Springer Berlin Heidelberg},
  address   = {Berlin, Heidelberg},
  pages     = {196--212},
  isbn      = {978-3-540-36415-3}
}

\appendices


\section{A Construction of Zero-bit Tamper-detection Code}\label{app:construction}
In this appendix we give an explicit information-theoretic construction of a zero-bit tamper-detection code. Our purpose is not to propose a practical scheme, but to show that the robustness threshold from Section~\ref{sec:limit-no-msg} is tight: there are codes that achieve soundness and tamper detection for any corruption rate strictly below the critical value. The construction is deliberately simple. The secret key itself serves as the codeword, and the decoder classifies a candidate string according to its Hamming distance from the key. This partitions the space of strings into three regions corresponding to $\valid$, $\tampered$, and $\invalid$, separated by a distance threshold $t$.

We now formalize this construction.

\begin{construction}\label{constr:no-msg}
Let $\cdwlen$ be the codeword length, and fix a parameter $\delta \in (0,1)$. Define $\thr = \cdwlen(1-\frac{1}{q})(1-\delta)$. The code $\Gamma = (\KGen, \Encode, \Decode)$ is defined as:
\begin{description}
    \item[Key Generation $\KGen(1^\secpar)$:] Sample a secret key $\sk \getsr \Sigma^{\cdwlen}$ uniformly at random, where $|\Sigma| = q \ge 2$.
    \item[Encoding $\Encode(\sk)$:] Output the codeword $\cdw = \sk$.
    \item[Decoding $\Decode(\sk, \cdw)$:] Given $\cdw \in \Sigma^{\cdwlen}$:
        \begin{itemize}
            \item If $\dist(\cdw, \sk) > \thr$, output $\invalid$.
            \item If $0 < \dist(\cdw, \sk) \leq \thr$, output $\tampered$.
            \item If $\cdw = \sk$, output $\valid$.
        \end{itemize}
\end{description}
\end{construction}

\noindent
The correctness property for this construction is immediate: the decoder always accepts the honest codeword, which is just the secret key.

What remains are the two core security properties: soundness, and robust detection of tampering up to the optimal threshold. Both follow from a simple analysis of Hamming distance and concentration.

\begin{thm}[Security of the Simple Construction]\label{thm:sec-constr-no-msg}
Let $\delta \in (0,1)$ be a \blue{constant}, and let $\cdwlen$ and $q = |\Sigma|$ such that $\frac{\cdwlen}{q}\in \omega(\log \secpar)$. \footnote{This can always be enforced by increasing the codeword length $\cdwlen$ with respect to the cardinality $q$ of $\Sigma$.} The zero-bit tamper-detection code $\Gamma$ from~\Cref{constr:no-msg} satisfies soundness and $\cF_{\alpha \cdwlen}$-tamper detection for $\alpha = (1-\frac{1}{q})(1-\delta)$.
\end{thm}

\begin{proof}
We prove each property in turn.

\textbf{Soundness:}
Fix any string $\hat{\cdw} \in \Sigma^{\cdwlen}$ and random key $\sk \getsr \Sigma^{\cdwlen}$.  
Let $X$ be the number of positions where $\hat{\cdw}[i] = \sk[i]$; $X$ is binomial with mean $\cdwlen/q$.  
The decoder outputs $\invalid$ unless
\[
X >  \cdwlen - \thr = \cdwlen - \cdwlen\Bigl(1-\frac{1}{q}\Bigr)(1-\delta) \geq \frac{\cdwlen}{q}(1+\delta),
\]
where the last inequality holds whenever $q \geq 2$.

By the Chernoff bound,
\[
\Prob\left[X \geq \frac{\cdwlen}{q}(1+\delta)\right] \leq \exp\left(-\frac{\delta^2 \cdwlen}{3q}\right),
\]
which is negligible whenever $\frac{\cdwlen}{q} \in \omega(\log \secpar)$ and $\delta > 0$ is constant.

\textbf{Tamper Detection:}
Fix any $f \in \cF_{\alpha \cdwlen}$ with $\alpha = (1-\frac{1}{q})(1-\delta)$.  
Let $\cdw = \sk$ and $\tilde{\cdw} = f(\cdw)$, with $\tilde{\cdw} \neq \cdw$.  
By definition of $f$, $\dist(\tilde{\cdw}, \sk) \leq \cdwlen(1-\frac{1}{q})(1-\delta) = \thr$.  
Thus, the decoder outputs $\tampered$.

Both properties hold as claimed.
\end{proof}

\end{document}